\documentclass[11pt,onecolumn]{article}
\usepackage[top=1in, bottom=1in, left=1.25in, right=1.25in]{geometry}
\usepackage{appendix}
\usepackage{amsfonts}
\usepackage{amsmath}
\usepackage{amssymb}
\usepackage{bm}
\usepackage{graphicx}
\usepackage{color, soul}
\usepackage{stfloats}
\usepackage{algorithm}
\usepackage{algorithmic}
\usepackage{booktabs}
\DeclareMathOperator*{\argmin}{argmin}
\usepackage{subfigure}
\usepackage{hyperref}

\usepackage[amsmath,thmmarks]{ntheorem}
\usepackage{theorem}

\newtheorem{theorem}{Theorem}
\newtheorem{lemma}{Lemma}
\newtheorem{remark}{Remark}

\newtheorem{definition}{Definition}

\theoremheaderfont{\sc}\theorembodyfont{\upshape}
\theoremstyle{nonumberplain}
\theoremseparator{}
\theoremsymbol{\rule{1ex}{1ex}}
\newtheorem{proof}{Proof}

\title{Cross Validation in Compressive Sensing and its Application of OMP-CV Algorithm}
\author{Jinye Zhang, Laming Chen, Petros T. Boufounos, and Yuantao Gu \thanks{This work was partially supported by the National Program on Key Basic Research Project (973 Program 2013CB329201) and the National Natural Science Foundation of China (NSFC 61371137). PTB is exclusively supported by Mitsubishi Electric Research Laboratories. The corresponding author of this paper is Yuantao Gu (gyt@tsinghua.edu.cn).}}
\date{~}
\begin{document}
\maketitle
\begin{abstract}
Compressive sensing (CS) is a data acquisition technique that measures sparse or compressible signals at a sampling rate lower than their Nyquist rate. Results show that sparse signals can be reconstructed using greedy algorithms, often requiring prior knowledge such as the signal sparsity or the noise level. As a substitute to prior knowledge, cross validation (CV), a statistical method that examines whether a model overfits its data, has been proposed to determine the stopping condition of greedy algorithms. This paper first analyzes cross validation in a general compressive sensing framework and developed general cross validation techniques which could be used to understand CV-based sparse recovery algorithms. Furthermore, we provide theoretical analysis for OMP-CV, a cross validation modification of orthogonal matching pursuit, which has very good sparse recovery performance. Finally, numerical experiments are given to validate our theoretical results and investigate the behaviors of OMP-CV.

\end{abstract}

\section{Introduction}
%introduction of Compressive Sensing and Cross Validation
Compressive sensing (CS) is a new data acquisition technique that aims to measure sparse and compressible signals at sampling rate smaller than the Nyquist rate \cite{donoho2006compressed} \cite{candes2006compressive}. Its fundamental promise is that some certain classes of signals, such as natural images, have a sparse representation where most of the coefficients are approximately zero. Generally speaking, compressive sensing consists of two main building blocks: sampling an $N$-dimensional $k$-sparse signal by computing its $M$ (which is much smaller than $N$) linear projections and reconstructing the signal using various recovery methods.

Although the reconstruction of the original signal $\mathbf{x}$ from its $M$ measurements is an ill-posed problem, it can be achieved by using the prior knowledge that $\mathbf{x}$ is sparse, i.e. $k\ll N$. One important result in CS theory is that $\mathbf{x}$ can be reconstructed using optimization strategies aiming to find the sparsest signal matching with the measurements, which can be viewed as an $l_0$ norm minimization problem  \cite{venkataramani1998sub}. Although the $l_0$ minimization is NP-hard, it was demonstrated \cite{candes2005decoding} that it is equivalent to an $l_1$ optimization problem as long as the sensing matrix satisfies the so-called Restricted Isometry Property (RIP) with a constant parameter. In addition, such $l_1$ optimization problem could be solved efficiently via linear programming (LP) techniques.

Apart from the LP techniques, a family of iterative greedy algorithms also received significant attention due to their low computational complexity. They include OMP, ROMP, StOMP, SP, and CoSaMP \cite{tropp2007signal,needell2009uniform,donoho2012sparse,dai2009subspace,needell2009cosamp}. The basic idea behind these algorithms is to find the support set of the signal iteratively. At each iteration, one or several coordinates of the vector $\mathbf{x}$ are selected into the current support set based on the correlation between the columns of sensing matrix and the measurement residual. With $k$ iterations needed, the computational complexity of OMP, for example, is roughly $O(k m N)$ \cite{dai2009subspace}. 

Most greedy algorithms require prior knowledge such as sparsity or noise level to properly stop the iteration, which however could not be satisfied in most practical cases. Without such information the termination of the algorithm may be too early or too late. In the former case, the signal will not be completely recovered (underfitting), while in the latter case some portion of the noise will be treated as signal (overfitting). In both cases, therefore, the reconstruction quality may greatly deteriorate. 

As a substitute to prior knowledge, cross validation (CV) was proposed \cite{boufounos2007sparse} to determine the stopping condition of greedy algorithms. Cross validation \cite{devijver1982pattern, geisser1993predictive, kohavi1995study, picard1984cross, shao1993linear} is a statistical technique that separates a data set into a training (estimation) set and a testing (cross validation) set. The training set is used to construct the model and the testing set is used to adjust the model order so that the noise is not overfitted. When CV is utilized in compressive sensing, the measurement vector is split into a reconstruction measurement and a cross validation measurement. The former is used to reconstruct the sparse signal using a greedy algorithm, while the latter is to decide the stopping condition. The basic idea behind this technique is to sacrifice a small amount of measurements in exchange of prior knowledge. In a nutshell, this technique makes it possible for greedy algorithms to reconstruct the signal without prior knowledge like sparsity or noise level.

\subsection{Related Works}
The idea of applying cross validation in compressive sensing is first proposed by Boufounos, Duarte, and Baraniuk in \cite{boufounos2007sparse}, where the general framework of CS-CV modification is founded. In a CS-CV modified algorithm, both the sensing matrix and the measurement vector are separated into the reconstruction part and the cross validation part. When the former part is utilized iteratively to construct the support set, the latter is adopted to calculate the CV residual and to determine the stopping condition. As soon as the CV residual is smaller than a given constant, its corresponding recovered signal will therefore be outputted as the reconstructed signal. The above work is the first step that introduced cross validation into the field of compressive sensing. 

Another important work in CS literature related to cross validation is made by Ward \cite{ward2009compressed}, who cleverly used the Johnson-Lindenstrauss (JL) Lemma to evaluate the recovery status. In this work, the reconstruction matrix is used for recovering the sparse signal and the cross validation matrix is used for estimating the reconstruction error. The dependence of the desired estimation accuracy and the confidence level in the prediction on the number of CV measurements is also studied. The above work offers us a tractable way for parameter selection in sparse recovery algorithms using CV for recovery error estimation. 

\subsection{Our Contribution}
The main contribution of this work is two-fold. We first develop some general cross validation techniques for compressive sensing. Then the theoretical analysis on OMP-CV algorithm is conducted comprehensively. 

The general cross validation techniques we provide could basically answer the following two questions, referred to as \emph{general CV problems} in the reminder of this paper.

\begin{enumerate}
\item (Recovery error estimation) Given a reconstructed signal, with what
  accuracy and what probability could its CV residual provide bounds
  on its recovery error?
\item (Recovery error comparison) Given a pair of reconstructed signals, with
  what probability could the comparison between their CV residuals
  correctly evaluate their recovery errors?
\end{enumerate}
To solve these problems, we first calculate the probability distribution of CV residuals. Consequently, by transforming the distribution into inequalities that hold with certain probability, we directly answer the above two questions.

Equipped with the general cross validation techniques, we then analyze the OMP-CV algorithm. We refer to the algorithm output, which is the reconstructed signal with the smallest CV residual, as the \emph{OMP-CV output}. The reconstructed signal with the smallest recovery error is referred to as the \emph{oracle output}. 

Our analysis result shows that the recovery error of the OMP-CV output is very close to that of the oracle output with high probability, given that the oracle output recovers all indices in the support set of the original signal. In order to achieve the above result, we first analyze the internal structure of two recovered signals in different OMP iterations. We then study how their CV residuals affect the recovery errors using the techniques developed for the general CV problems. Finally we generalize the recovery error comparison between two recovered signals, which are generated in different OMP iterations, to that of all recovered signals. Therefore we could estimate how close the OMP-CV output is to the oracle output.

The reminder of the paper is organized as follows. Section \ref{section preliminaries} contains the problem formulation, OMP-CV description, and some mathematical tools required in the following analysis. Section \ref{section cross validation} analyzes the cross validation techniques for compressive sensing while Section \ref{section OMP-CV} provides a comprehensive discussion on OMP-CV algorithm. Numerical simulations are given in Section \ref{section Simulation} to verify the theoretical content. Concluding remarks are drawn in Section \ref{section conclusion}, while proofs of some theorems are presented in the Appendix.

\section{Preliminaries \label{section preliminaries}}

\subsection{Notations \label{problem formulation}}

We consider an unknown $k$-sparse signal $\mathbf{x} \in \mathbb{R}^N$ observed using $M$ linear measurements corrupted by additive noise. Let $T$ be the support set of $\bf x$ and $\vert T \vert = k$ to
denote the cardinality of $T$. The vector $\mathbf{x}_T$ contains the elements of $\mathbf{x}$ indexed by $T$. To implement the CV-based modification, we separate the original $M$ by $N$ sensing matrix to a reconstruction matrix $\mathbf{A} \in \mathbb{R}^{m \times N}$ and a CV matrix $\mathbf{A}_{\rm{cv}} \in \mathbb{R}^{m_{\rm{cv}} \times   N}$. The measurement vector is also separated accordingly, to a reconstruction measurement $\mathbf{y} \in \mathbb{R}^m$ and a CV measurement $\mathbf{y}_{\rm{cv}} \in \mathbb{R}^{m_{\rm{cv}}}$. In this paper we only consider Gaussian sensing matrices and additive Gaussian noises. The reconstruction matrix $\mathbf{A}$ is properly normalized to have unit column norm. Because the same data acquisition system is assumed to be used to obtain both the reconstruction and CV measurements, the CV matrix $\mathbf{A}_{\rm{cv}}$ is normalized to have column norm equal to $\sqrt{m_{\rm{cv}} / m }$ and the CV noise has the same per measurement variance as the measurement noise. In other words, the notations can be formulated as
\begin{align*}
\mathbf{y} &= \mathbf{A} \mathbf{x} +\mathbf{n}, &  \mathbf{n} &= \sigma_{\rm{n}} \mathbf{a}_{\rm{n}},\\
\mathbf{y}_{\rm{cv}} &= \mathbf{A}_{\rm{cv}} \mathbf{x} + \mathbf{n}_{\rm{cv}}, & \mathbf{n}_{\rm{cv}} &= \sigma_{\rm{n}} \mathbf{a}_{{\rm{cv}},{\rm{n}}} ,
\end{align*}
where the entries of $\mathbf{A}$, $\mathbf{A}_{\rm{cv}}$, $\mathbf{a}_{\rm{n}}$, and $\mathbf{a}_{{\rm{cv}},{\rm{n}}}$ are i.i.d normally distributed with mean zero and variance $1/m$.

To make the analysis more clear, we emphasize that in this paper, the input signal is considered as deterministic while the sampling matrix and noise are random. Without loss of clarity and for notational simplicity, the random variables and their realizations are denoted by same notation. 

\subsection{Orthogonal Matching Pursuit with Cross Validation}

\begin{table}[t]
\renewcommand{\arraystretch}{1.2}
\caption{OMP-CV Algorithm}
\label{OMP-CV Algorithm} 
\begin{center}
\begin{tabular}{l}
\toprule[1pt]
{\bf Input:} \hspace{0.5em} $\mathbf{A}, \mathbf{A}_{\rm{cv}}, \mathbf{y}, \mathbf{y}_{\rm{cv}}, d$;\\
{\bf Output:} \hspace{0.5em} $\mathbf{\hat{x}}$.\\
\hline
{\bf Initialization:} \hspace{0.5em} Set $p=1$, $\epsilon_{\rm{cv}}^0 = \Vert \mathbf{y}_{\rm{cv}}\Vert_2^2$;\\
{\bf Repeat:}\\
\hspace{1.5em} Compute $\mathbf{\hat{x}}^p$ using an OMP iteration;\\
\hspace{1.5em} Compute $\epsilon_{\rm{cv}}^p = \Vert \mathbf{A}_{\rm{cv}} \mathbf{\hat{x}}^p - \mathbf{y}_{\rm{cv}} \Vert_2^2$;\\
\hspace{1.5em} Increment $p$ by 1; \\
{\bf Until:} \hspace{0.5em} $p \geq d$\\
Compute $o_{\rm{cv}} = \displaystyle \argmin_p \epsilon_{\rm{cv}}^p$; \\
{\bf Return:} $\mathbf{\hat{x}} = \mathbf{\hat{x}}^{o_{\rm{cv}}}$. \\
\bottomrule[1pt]
\end{tabular}
\end{center}
\end{table}

The OMP-CV algorithm proposed in \cite{boufounos2007sparse} is a noise- and sparsity-robust greedy recovery algorithm that adopts CV in OMP. In this algorithm, every iteration can be viewed as two separate parts: reconstructing the signal by OMP and evaluating the recovered signal by cross validation techniques, which is utilized to properly terminate the iteration before the recovery starts to overfit the noise. The OMP-CV algorithm that studied in this work is slightly different from its original version. One may refer to Table \ref{OMP-CV Algorithm}, where the iteratively reconstructed signal is chosen as  output based on the criteria that its CV residual is the smallest rather than less than a certain constant. 

In Table \ref{OMP-CV Algorithm}, we use $\mathbf{\hat{x}}^p$ and $T^p$ to denote the recovered signal and its support, respectively, in the $p$-th iteration. The difference between the recovered signal
$\mathbf{\hat{x}}^p$ and the input signal $\mathbf{x}$ is denoted using $\Delta\mathbf{x}^p$. The recovery error and the CV residual corresponding to $\mathbf{\hat{x}}^p$ are denoted by $\varepsilon^p_{\rm x}$ and $\epsilon_{\rm{cv}}^p$, respectively.
\begin{equation*}
\Delta\mathbf{x}^p \triangleq \mathbf{x} - \mathbf{\hat{x}}^p, \qquad   \varepsilon^p_{\rm x} \triangleq \Vert \Delta\mathbf{x}^p \Vert_2^2, \qquad
\epsilon_{\rm{cv}}^p \triangleq \Vert \mathbf{y}_{\rm{cv}} - \mathbf{A}_{\rm{cv}} \mathbf{\hat{x}}^p \Vert _2^2. 
\end{equation*}

Here we are trying to present some intuition about how OMP-CV works. Please refer to Figure \ref{fig OMP-CV behavior}, which demonstrates the evolution by iteration of residual, CV residual, and recovery error. One may notice that the trend of residual in iteration behaves abruptly different comparing to that of recovery error, as soon as the reconstructed signal starts to overfit the noise. Therefore, residual fails to serve as an indicator for correctly terminating the algorithm. However, the CV residual evolves similarly as that of the recovery error. This is the reason that CV modification could improve the performance of OMP and other related greedy algorithms. 

\begin{figure}
\begin{center}
\includegraphics[scale=0.6]{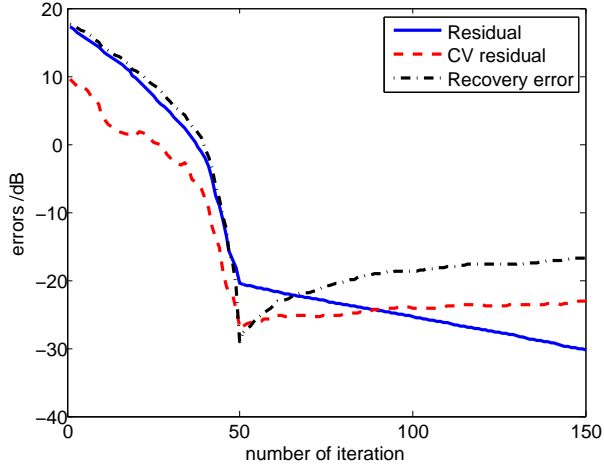}
\caption{The evolution of residual, CV residual, and recovery error of OMP-CV. \label{fig OMP-CV behavior}}
\end{center}
\end{figure}

We would like to emphasize that OMP-CV is a highly practical algorithm. OMP-CV does not require prior information such as noise level or sparsity. Instead, only the maximum number of iterations is required as input\footnote{One may notice that the number of maximum iteration cannot be greater than that of the measurements, because regular OMP will produce zero residual after that and conclude.}. Furthermore, OMP-CV provides an estimate of the recovery error in its CV residual. This is very helpful because it could be immediately detected if the algorithm did not recover the signal well. Finally, it is worthwhile to mention that by properly setting $m_{\rm{cv}}$, the recovery performance of OMP-CV competes with that of OMP even when the accurate information of noise level is given for the latter. These advantages will be supported both theoretically and empirically in following part of this work.

\subsection{Restricted Isometry Property}

In compressive sensing literature, a large body of works focuses on the theoretical analysis of sparse recovery algorithm performance. Among these theoretical analysis, the RIP \cite{candes2006near} becomes one of the most helpful and widely used tools. It quantifies the idea that the geometry of sparse signals should be preserved under the mapping of sensing matrix. In this paper, the RIP is frequently used to study the internal structure of reconstructed signals acquired in different iterations of OMP.  

\begin{definition} \cite{dai2009subspace}\label{definition RIP}
(RIP):  A sensing matrix $\mathbf{A} \in \mathbb{R}^{m \times N}$ is said to satisfy the Restricted Isometry Property with parameters $(k,\delta)$ for $k \leq m$, $0 \leq \delta \leq 1$, if for all index sets $T\subset \{1, \cdots , N \}$ such that $\vert T \vert \leq k$ and for all $\mathbf{x} \in \mathbb{R}^{N}$, one has 
\begin{equation}
(1-\delta)\Vert \mathbf{x}_{T} \Vert_2^2 \leq \Vert \mathbf{A}_T \mathbf{x}_T \Vert_2^2 \leq (1+\delta) \Vert \mathbf{x}_T \Vert_2^2.
\end{equation}
We define $\delta_k$, the Restricted Isometry Constant (RIC), as the infimum of all parameters $\delta$ for which the RIP holds.
\end{definition}

\begin{remark} \cite{dai2009subspace}
(RIP and eigenvalues):If a sensing matrix $\mathbf{A}\in \mathbb{R}^{m \times N}$ satisfies the RIP with parameters $(k,\delta_k)$, then for all $T \subset \{1,\cdots,N\}$ such that $\vert T \vert \leq k$, it holds that 
\begin{equation}
1-\delta_k \leq \lambda_{\min}(\mathbf{A}_T' \mathbf{A}_T) \leq \lambda_{\max}(\mathbf{A}_T' \mathbf{A}_T)\leq 1+\delta_k,
\end{equation}
where $(\cdot)', \lambda_{\min}(\cdot)$, and $\lambda_{\max}(\cdot)$ denote the transpose, the minimal and maximal eigenvalues of a matrix, respectively.
\end{remark}

Most known families of matrices satisfying the RIP with optimal or near-optimal performance guarantees are random. Specifically, in this paper we will deal with Gaussian random matrix, which has the following property. 

\begin{remark} \cite{needell2009cosamp}
(Matrices satisfying the RIP): If the entries of $\sqrt{m} \mathbf{A}$ are independent and identically distributed standard normal variables and if 
\begin{equation}
m \geq C {k \log (N/k)}/{\tau^2},
\end{equation}
where $C$ is a constant, then, one has $\delta_k \leq \tau$
except with probability $N^{-1}$.
\end{remark}

If a sensing matrix satisfies the RIP, it has some other properties that are required in our analysis. 
The first one is a simple translation of Definition \ref{definition RIP}. 

\begin{lemma} \label{RIP lemma 1}
Suppose $\mathbf{A}$ has an RIP of $\delta_k$ and $T$ is a set of $k$ indices or fewer. For all $\mathbf{x} \in \mathbb{R}^{N}$, one has 
\begin{gather}
\sqrt{(1-\delta_k)}\Vert \mathbf{x}_T \Vert_2 \leq \Vert \mathbf{A}_T \mathbf{x}_T \Vert_2 \leq \sqrt{(1+\delta_k)} \Vert \mathbf{x}_T \Vert_2,\\
(1-\delta_k)\Vert \mathbf{x}_T \Vert_2 \leq \Vert \mathbf{A}'_T \mathbf{A}_T \mathbf{x}_T \Vert_2 \leq (1+\delta_k) \Vert \mathbf{x}_T \Vert_2,\\
\frac{1}{(1+\delta_k)}\Vert \mathbf{x}_T \Vert_2 \leq \Vert (\mathbf{A}'_T \mathbf{A}_T)^{-1} \mathbf{x}_T \Vert_2 \leq \frac{1}{(1-\delta_k)} \Vert \mathbf{x}_T \Vert_2.
\end{gather}
\end{lemma} 

A second consequence is that the disjoint sets of columns from the sensing matrix span nearly orthogonal subspaces. To quantify this observation, we have the following results.

\begin{lemma} \cite{needell2009cosamp} \label{RIP lemma 2}
(Approximate orthogonality): Suppose $\mathbf{A}$ has a RIC of $\delta_{\vert S \vert + \vert T \vert}$, where $S,T \subset \{1,2, \cdots ,N\}$ are disjoint sets. One has 
\begin{equation}
\Vert \mathbf{A}_S' \mathbf{A}_T \Vert_2 \leq  \delta_{\vert S \vert + \vert T \vert}.
\end{equation}
\end{lemma}

\begin{lemma} \cite{needell2009cosamp} \label{RIP lemma 3}
Let $S,T \subset \{1,2, \cdots ,N\}$ be two disjoint sets and suppose that $\delta_{\vert S \vert + \vert T \vert} < 1$. For all $\mathbf{x} \in \mathbb{R}^{N}$, one has
\begin{equation}
\Vert \mathbf{A}_S' \mathbf{A}_T \mathbf{x}_T\Vert_2 \leq  \delta_{\vert S \vert + \vert T \vert} \Vert \mathbf{x}_T\Vert_2.
\end{equation}
\end{lemma}

Based on the Lemma \ref{RIP lemma 1} and Lemma \ref{RIP lemma 3}, we derived the following lemma.
\begin{lemma} \label{RIP lemma 4}
Let $S,T \subset \{1,2, \cdots ,N\}$ be two disjoint sets and suppose that $\delta_{\vert S \vert + \vert T \vert} < 1$.  For all $\mathbf{x} \in \mathbb{R}^{N}$, one has
\begin{equation}
\Vert \mathbf{A}_S^ \dagger  \mathbf{A}_T \mathbf{x}_T\Vert_2 \leq  \frac{\delta_{\vert S \vert + \vert T \vert}}{1 - \delta_{\vert S \vert}} \Vert \mathbf{x}_T\Vert_2.
\end{equation}
\end{lemma}

\begin{proof}
The proof of Lemma \ref{RIP lemma 4} is postponed to Appendix \ref{app RIP lemma 4}.
\end{proof}

Before the statement of the last property, we make a definition of an orthogonal projection operator $P_T$.

\begin{definition}\label{definition P}
Let $\mathbf{A} \in \mathbb{R}^{m\times N}$ and $T \subset \{1,\cdots,N\}$. Define an operator $P_T$ as 
\begin{equation}
P_T \triangleq {\bf I} - \mathbf{A}_T \mathbf{A}_T^\dagger,
\end{equation}
where $(\cdot)^\dagger$ denotes the pseudo-inverse of a matrix.
\end{definition}

\begin{remark}
$P_T$ is an orthogonal projection operator whose function is to remove the component of a vector that is in the space spanned by $\mathbf{A}_T$. 
\end{remark}

\begin{lemma} \cite{davenport2009compressive} \label{RIP lemma 5}
Let $S,T \subset \{1,2, \cdots ,N\}$ be two disjoint sets and suppose that $\delta_{\vert S \vert + \vert T \vert} < 1$. For all $\mathbf{x} \in \mathbb{R}^{N}$, one has
\begin{equation}
\sqrt{1-\left( \frac{\delta_{\vert S \vert + \vert T \vert}}{1-\delta_{\vert S \vert + \vert T \vert}}\right)^2} \Vert \mathbf{A}_T \mathbf{x}_T \Vert_2 \leq \Vert P_S \mathbf{A}_T \mathbf{x}_T \Vert_2  \leq \Vert \mathbf{A}_T \mathbf{x}_T \Vert_2.
\end{equation}
\end{lemma}

\section{Cross Validation in Compressive Sensing \label{section cross validation}}
The results of general cross validation techniques in compressive sensing are presented in this section. These techniques will be utilized in the later analysis of OMP-CV algorithm. We would like to emphasize that besides the application of analyzing OMP-CV algorithm, the content in this section is fundamentally general and can be used to understand other CV-based sparse recovery algorithms as well. 

\subsection{Recovery Error Estimation}

First, we start with calculating the probability distribution of $\epsilon_{\rm{cv}}$, which is described by the following lemma.

\begin{lemma}\label{cv lemma 1}
Let $\mathbf{\hat{x}}$ be the recovered signal. Assuming there are enough measurements for cross validation, one has 
\begin{equation}
\epsilon_{\rm{cv}} = \Vert \mathbf{y}_{\rm{cv}} - \mathbf{A}_{\rm{cv}} \mathbf{\hat{x}} \Vert _2^2 \sim \mathcal{N} (\mu, \sigma^2),
\end{equation}
where $\mu = \frac{m_{\rm{cv}}}{m} ( \varepsilon_{\rm{x}}+\sigma_{\rm{n}}^2 ), \sigma^2 = \frac{2 m_{\rm{cv}}}{m^2} (\varepsilon_{\rm{x}}+\sigma_{\rm{n}}^2)^2$, and $\varepsilon_{\rm x} = \Vert \mathbf{x} - \mathbf{\hat{x}} \Vert_2^2$. ${\bf y}_{\rm cv}, {\bf A}_{\rm cv}, m, m_{\rm cv}$, and $\sigma_{\rm n}^2$ are defined in section \ref{problem formulation}.
\end{lemma}

\begin{proof}
The proof is postponed to Appendix \ref{app cv lemma 1}.
\end{proof}

The condition ``there are enough measurements for cross validation'' is required in one step of the proof of Lemma \ref{cv lemma 1}, which is the probability distribution approximation via Central Limit Theorem (CLT). According to Central Limit Theorem, the real probability distribution of $\epsilon_{\rm{cv}}$ converges absolutely to the above approximated result with the increase of $m_{\rm{cv}}$ and the approximation error becomes negligible very fast. In fact, the approximation is rather good when $m_{\rm{cv}}$ is greater than tens. Considering that compressive sensing always deals with large scale problems, such condition could be readily satisfied. This lemma will be verified by simulation result in section \ref{simulation 1}. 

One immediate consequence of Lemma \ref{cv lemma 1} is that we can use $\epsilon_{\rm{cv}}$ to provide estimation for $\varepsilon_{\rm{x}}$ in term of an inequality that holds with certain probability. 

\begin{theorem}\label{cv theorem 1}
(CV estimation): Assuming there are enough measurements for cross validation, the following inequality holds with probability $ \rm{erf} ({\lambda}/{\sqrt{2}})$
\begin{equation}
h(\lambda,+)\epsilon_{\rm{cv}} - \sigma_{\rm{n}}^2 \leq \varepsilon_{\rm{x}} \leq h(\lambda,-)\epsilon_{\rm{cv}} -\sigma_{\rm{n}}^2,
\end{equation}
where
\begin{equation}
h(\lambda,\pm) \triangleq \frac{m}{m_{\rm{cv}}} \frac{1}{1 \pm \lambda \sqrt{{2}/{m_{\rm{cv}}}}},
\end{equation} 
$\lambda$ is a parameter concerning the trade off between the probability and the estimation accuracy, and ${\rm erf} (u)\triangleq \frac{1}{\sqrt{\pi}} \int_{-u}^{u} {\rm e}^{-t^2} {\rm d} t$ denotes the error function of standard Gaussian distribution.
\end{theorem}

\begin{proof}
According to Lemma \ref{cv lemma 1} and the properties of Gaussian distribution, the result can by derived after some basic algebra.
\end{proof}

Theorem \ref{cv theorem 1} basically solves the first general CV problem that one can estimate the interval of recovery error $\varepsilon_{\rm{x}}$ by the observed CV residual $\epsilon_{\rm{cv}}$ with probability $ \rm{erf} ({\lambda}/{\sqrt{2}})$. One may notice that the width of the interval is 
\begin{equation}
\frac{m}{m_{\rm{cv}}} \frac{2\lambda\sqrt{2}}{\sqrt{m_{\rm{cv}}} - {2\lambda^2}/{ \sqrt{m_{\rm{cv}}}}}\epsilon_{\rm{cv}},
\end{equation}
which means that the bounds become tighter with the increase of the number of measurements used for cross validation. In particular, if $m_{\rm{cv}}$ is far larger than $2\lambda^2$, one could accurately estimate $\varepsilon_{\rm{x}}$ by the estimator $\frac{m}{m_{\rm{cv}}}\epsilon_{\rm{cv}}-\sigma_{\rm n}^2$. For example, with $\lambda = 3$, $m = 400$, and $m_{\rm{cv}} = 80$, it holds with probability $99.4 \%$ that 
\begin{equation}
3.4\epsilon_{\rm{cv}} - \sigma_{\rm{n}}^2 \leq \varepsilon_{\rm{x}} \leq 9.5\epsilon_{\rm{cv}} -\sigma_{\rm{n}}^2.
\end{equation}

\subsection{Recovery Error Comparison}

In the second general CV problem, we try to compare two recovered signals, $\mathbf{\hat{x}}^p$ and $\mathbf{\hat{x}}^q$. According to Theorem \ref{cv theorem 1}, if $\varepsilon_{\rm{x}}^p$ is larger than $\varepsilon_{\rm{x}}^q$, the CV residuals should also have $\epsilon_{\rm{cv}}^p > \epsilon_{\rm{cv}}^q$ with high probability. Therefore, we could be able to compare the recovery errors by simply comparing their CV residuals. This section presents the mathematical formulation of this probability. 

For simplicity, we have the following definition.
\begin{definition} \label{Generalized sensing matrix and input signal}
(Generalized reconstrction matrix and input signal): 
Let $\mathbf{A}_{\rm{g}} \triangleq [\mathbf{A},\mathbf{a}_{\rm{n}}]$ and $\mathbf{x}_{\rm{g}} \triangleq [\mathbf{x}',\sigma_{\rm{n}}]'$, then 
\begin{equation}
\mathbf{y} = \mathbf{A} \mathbf{x} +\mathbf{n} = \mathbf{A}_{\rm{g}} \mathbf{x}_{\rm{g}},
\end{equation}
where $\mathbf{A}_{\rm{g}}$ is called the \emph{generalized reconstruction matrix}, and $\mathbf{x}_{\rm{g}}$ is called the \emph{generalized input signal}.
\end{definition}
The purpose of making this definition is to simplify our analysis. Compared to the original input signal $\mathbf{x}$, $\mathbf{x}_{\rm{g}}$ has an extra term which represents the noise and can never be recovered. It can be understood as a generalized version of input signal with a part that reflects the measurement noise. The generalized versions of $\Delta\mathbf{x}^p$ and $\varepsilon^p_{\rm x}$ are $\Delta\mathbf{x}_{\rm{g}}^p$ and $\varepsilon_{\rm{g}}^p$, respectively. 
\begin{equation*}
\Delta\mathbf{x}_{\rm{g}}^p \triangleq [{(\Delta \mathbf{x}^p)}',\sigma_{\rm{n}}]', \qquad \varepsilon_{\rm{g}}^p\triangleq \Vert \Delta \mathbf{x}_{\rm{g}}^p \Vert_2^2.
\end{equation*}

We start with calculating the probability distribution of $\Delta \epsilon_{\rm{cv}} = \epsilon_{\rm{cv}}^p - \epsilon_{\rm{cv}}^q$.

\begin{lemma}\label{cv lemma 2}
Let $\mathbf{\hat{x}}^p$ and $\mathbf{\hat{x}}^q$ be two recovered signals.
Assuming there are enough measurements for cross validation, one has
\begin{equation}
\Delta \epsilon_{\rm{cv}} = \epsilon_{\rm{cv}}^p - \epsilon_{\rm{cv}}^{q} \sim \mathcal{N}( \mu, \sigma^2),
\end{equation}
where $\mu = \frac{m_{\rm{cv}}}{m} ( \varepsilon_{\rm{g}}^p - \varepsilon_{\rm{g}}^q )$ and $\sigma^2 =\frac{2 m_{\rm{cv}}}{m^2} [(\varepsilon_{\rm{g}}^p)^2 + (\varepsilon_{\rm{g}}^q)^2 - 2\rho_{\rm{g}}^2 \varepsilon_{\rm{g}}^p \varepsilon_{\rm{g}}^q ]$, and 
\begin{equation}\rho_{\rm{g}} \triangleq \frac{\langle\Delta \mathbf{x}_{\rm{g}}^p, \Delta \mathbf{x}_{\rm{g}}^q\rangle}{\Vert \Delta \mathbf{x}_{\rm{g}}^p\Vert_2 \Vert\Delta \mathbf{x}_{\rm{g}}^q\Vert_2}
\end{equation}
denotes the correlation coefficient of the two generalized recovery error signals.
\end{lemma}

\begin{proof}
The proof of Lemma \ref{cv lemma 2} is deferred to Appendix \ref{app cv lemma 2}.
\end{proof}

With Lemma \ref{cv lemma 2}, we are much closer to the answer to the key question, which is presented in the following theorem.

\begin{theorem}\label{cv theorem 2}
(CV comparison) Let $\mathbf{\hat{x}}^p$ and $\mathbf{\hat{x}}^q$ be two recovered signals and assume there are enough measurements for cross validation. If $\varepsilon_{\rm{x}}^p \geq \varepsilon_{\rm{x}}^q$,  it holds with probability $\Phi(\lambda)$ that $\epsilon_{\rm{cv}}^p \geq \epsilon_{\rm{cv}}^q$, where $\lambda$ is determined by
\begin{equation} \label{cv theorem e2}
\frac{1}{\lambda^2} = \frac{2}{m_{\rm{cv}}} \left[ 1+2(1-\rho_{\rm{g}}^2)\frac{\varepsilon_{\rm{g}}^p \varepsilon_{\rm{g}}^q}{(\varepsilon_{\rm{g}}^p - \varepsilon_{\rm{g}}^q)^2}\right],
\end{equation}
and $\Phi (u) \triangleq \frac{1}{\sqrt{2\pi}}\int_{-\infty}^u {\rm e}^\frac{-t^2}{2} {\rm d}t$ is the cumulative distribution function of standard Gaussian distribution.
\end{theorem}

\begin{proof}
According to Lemma \ref{cv lemma 2} and the properties of Gaussian distribution, the result can by derived after some basic algebra.
\end{proof}

$\Phi(\lambda)$ is called \emph{CV comparison success probability}, because it is the probability that the size order of CV residuals correctly evaluates that of the recovery errors. Next we discuss the parameters that determines $\lambda$.  

\begin{itemize}
\item Parameter $m_{\rm{cv}}$:\\
When $m_{\rm{cv}}$ increases, $\frac{1}{\lambda^2}$ decreases, thus $\lambda$ increases. Hence, a larger $m_{\rm{cv}}$ is required to obtain a larger $\lambda$. This corresponds to our intuition: $m_{\rm{cv}}$ is the number of CV measurements and the more CV measurements we have, the better CV performance, which is shown as a higher CV comparison success probability, we should be able to obtain. 
\item Parameter $\rho_{\rm{g}}$:\\
The value of $\lambda$ is in positive relation with $\rho_{\rm{g}}^2$. We would like to note that $\rho_{\rm{g}}$ describes the correlation of $\Delta \mathbf{x}_{\rm{g}}^p$ and $\Delta \mathbf{x}_{\rm{g}}^q$. As an interpretation, CV comparison success probability is often higher if the two recovered signals need to be compared are highly correlated. One possible explanation for this fact is that the similar part of the signal yields similar part in CV measurements and the randomness of the comparison problem comes from the dissimilar part of the recovered signal. From this observation we would like to emphasize that the CV comparison for two highly correlated signals is often of a greater success probability.
\item Parameter $\varepsilon_{\rm{g}}^p$ and $\varepsilon_{\rm{g}}^q$:\\
The influence of $\varepsilon_{\rm{g}}^p$ and $\varepsilon_{\rm{g}}^q$ is depicted by the following theorem: 
\end{itemize}

\begin{theorem} \label{cv theorem 3}
Let the CV measurement number $m_{\rm{cv}}$ and the similarity level $\rho_{\rm{g}}$ be fixed and assume there are enough measurements for cross validation, CV comparison success probability is equal to or higher than $\Phi(\lambda_0)$ if and only if the ratio of generalized signal error $\varepsilon_{\rm{g}}^p / \varepsilon_{\rm{g}}^q$ satisfies
\begin{equation}
\frac{\varepsilon_{\rm{g}}^p}{\varepsilon_{\rm{g}}^q} \geq 2C_0 + 1 + 2\sqrt{C_0^2 + C_0},
\end{equation} 
where $C_0 \triangleq \frac{\lambda_0^2(1-\rho_{\rm{g}}^2)}{m_{\rm{cv}}-2\lambda_0^2}$ is a constant related to $m_{\rm{cv}}$, $\lambda_0$ and $\rho_{\rm{g}}$.
\end{theorem} 
This Theorem can be derived from Lemma \ref{cv lemma 2} and Theorem \ref{cv theorem 2}. 

\begin{proof}
In this proof we just prove that if $\Phi (\lambda) \geq \Phi(\lambda_0)$, it holds that 
\begin{equation}
\frac{\varepsilon_{\rm{g}}^p}{\varepsilon_{\rm{g}}^q} \geq 2C_0 + 1 + 2\sqrt{C_0^2 + C_0}.
\end{equation}
The counterpart of the proof is extremely similar so we do not present it here. Since the CV comparison success probability is higher than $\Phi(\lambda_0)$ and the function $\Phi(u)$ is a monotonically increasing function, it holds that $\lambda > \lambda_0$. In addition, because $\lambda$ and $\lambda_0$ are both positive, then
\begin{equation}
\frac{1}{\lambda_0^2} \geq \frac{1}{\lambda^2} = \frac{2}{m_{\rm{cv}}} \left[1+2(1-\rho_{\rm{g}}^2)\frac{\varepsilon_{\rm{g}}^p \varepsilon_{\rm{g}}^q}{(\varepsilon_{\rm{g}}^p - \varepsilon_{\rm{g}}^q)^2}\right].
\end{equation}
Further
\begin{equation}
\frac{\varepsilon_{\rm{g}}^p}{\varepsilon_{\rm{g}}^q} + \frac{\varepsilon_{\rm{g}}^q}{\varepsilon_{\rm{g}}^p} \geq \frac{4\lambda_0^2}{m_{\rm{cv}} - 2\lambda_0^2}(1-\rho_{\rm{g}}^2) + 2.
\end{equation}
Because $\varepsilon_{\rm{g}}^p \geq \varepsilon_{\rm{g}}^q$, the left part of the above inequality monotonically increases with $\varepsilon_{\rm{g}}^p / \varepsilon_{\rm{g}}^q$. Thus 
\begin{equation}
\frac{\varepsilon_{\rm{g}}^p}{\varepsilon_{\rm{g}}^q} \geq \frac{2\lambda_0^2(1-\rho_{\rm{g}}^2)}{m_{\rm{cv}}-2\lambda_0^2} + 1 + \sqrt{\frac{4\lambda_0^4(1-\rho_{\rm{g}}^2)^2}{(m_{\rm{cv}}-2\lambda_0^2)^2} + \frac{4\lambda_0^2(1-\rho_{\rm{g}}^2)}{m_{\rm{cv}}-2\lambda_0^2}}.
\end{equation} 
Using the notation of $C_0$ we may complete the proof:
\begin{equation}
\frac{\varepsilon_{\rm{g}}^p}{\varepsilon_{\rm{g}}^q} \geq 2C_0 + 1 + 2\sqrt{C_0^2 + C_0}.
\end{equation}
where 
\begin{equation}
C_0 = \frac{\lambda_0^2(1-\rho_{\rm{g}}^2)}{m_{\rm{cv}}-2\lambda_0^2}.
\end{equation}
\end{proof}

\begin{remark}
The reason that $\rho_{\rm{g}}$ should be discussed can be explained as follows. For two recovered signals with a high correlation coefficient, their CV comparison success probability is often high. Hence, if we know the lower bound of $\rho_{\rm{g}}$ in advance, a much better CV recovery error comparison performance can be guaranteed. This is also the case of the OMP-CV algorithm. In OMP algorithm, a new index is incorporated into the current support set in each iteration and the recovered signal is determined by the current support set. Thus, the recovered signals of neighboring iterations have an extremely high correlation coefficient. As a result, the CV comparison success probability is very high when we compare the recovery errors of these recovered signals.
\end{remark}

The above analysis solves the recovery error estimation problem and the recovery error comparison problem. These results provide powerful tools for CV-based compressive sensing algorithms. Particularly, we studied the application of cross validation in OMP algorithm. The results are presented in following section.
\section{Sparsity and Noise Robust OMP \label{section OMP-CV}}
This section presents our theoretical analysis of the behavior of OMP-CV. First make the definitions:

\begin{definition} \label{definition alpha}
(Ratio of unrecovered signal and noise): The ratio $\alpha^p \in \mathbb{R}$, defined as\begin{equation}
\alpha^p \triangleq \frac{\Vert \mathbf{x}_{T \backslash T^p}\Vert_2}{\sigma_{\rm{n}}},
\end{equation}
measures to what extent the signal $\mathbf{\hat{x}}$ has not been recovered by $\mathbf{\hat{x}}^p$.
\end{definition}

\begin{definition}
(Oracle output and OMP-CV output): In OPM-CV, by \emph{oracle output} we mean the recovered signal that has the lowest recovery error. By \emph{OMP-CV output} we mean the recovered signal with the lowest CV residual.
\end{definition}

Intuitively, among recovered signals generated in different OMP iterations, the oracle output is most likely to be the OMP-CV output, which does not held all the time due to the randomness of the CV matrix. The following result describes relationship between the OMP-CV output and the oracle output.

\begin{theorem} \label{OMP-CV theorem 1}
In OMP-CV, assume that the oracle output is $\mathbf{\hat{x}}^o$ and $T \subset T^o$. For any recovered signal $\mathbf{\hat{x}}^p$ other than $\mathbf{\hat{x}}^o$: 
\begin{itemize}
\item if $T \backslash T^p \neq \emptyset$, then $\epsilon_{\rm{cv}}^o < \epsilon_{\rm{cv}}^p$ with probability $\Phi(\lambda)$, where 
\begin{equation}
\label{not properly recovered signal}
\lambda \geq \sqrt{\frac{m_{\rm{cv}}}{2}} \sqrt{ 1 - \rm{g}(\alpha^p) };
\end{equation}
\item if $T \backslash T^p = \emptyset$ and if $\mathbf{\hat{x}}^p$ is the OMP-CV output, then with probability greater than $\left\{ 1 - (d-k) [1 - \Phi(\lambda_0)] \right\}$ we have
\begin{equation}
\label{properly recovered signal}
\varepsilon_{\rm{g}}^p \leq C_1 {\varepsilon_{\rm{g}}^o},
\end{equation}
\end{itemize}
where $ \rm{g}(\alpha^p)$ is roughly proportional to $1 /(\alpha^p)^2$, $\lambda_0$ is a constant chosen to decide the probability with which (\ref{properly recovered signal}) holds, and $C_1$ is only related to $\lambda_0$ and $m_{\rm{cv}}$.
\end{theorem}

\begin{proof}
One may refer to Appendix \ref{app OMP-CV theorem 1} for the proof of Theorem \ref{OMP-CV theorem 1}. 
\end{proof}

\begin{figure}[h]
\begin{center}
\includegraphics[scale=0.6]{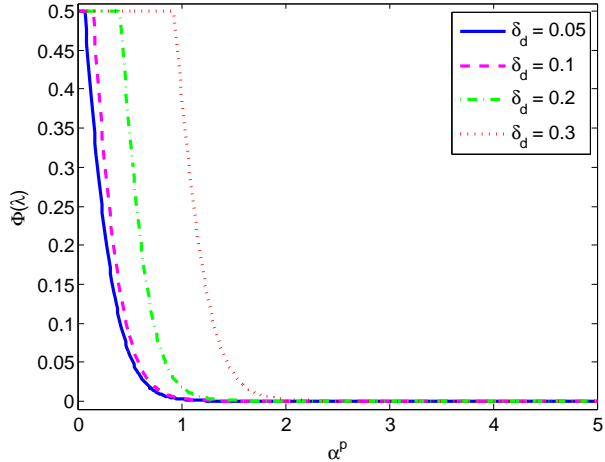}
\caption{This figure plots $\left[ 1-\Phi(\lambda) \right]$, the upper bound of the probability that $\epsilon_{\rm{cv}}^p < \epsilon_{\rm{cv}}^o$, with variation of $\alpha^p$ and RIC $\delta_d$. $\left[ 1-\Phi(\lambda) \right]$ decays sharply as $\alpha ^p$ increases. \label{fig cv success probability}}
\end{center}
\end{figure}

Theorem \ref{OMP-CV theorem 1} supports the recovery performance of OMP-CV. It divides recovered signals into two categories. For $\mathbf{\hat{x}}^p$ with $T \backslash T^p \neq \emptyset$, the
probability $\left[ 1-\Phi(\lambda) \right]$ decays sharply as $\alpha ^p$ increases. Since this is the upper bound of the probability that $\epsilon_{\rm{cv}}^p < \epsilon_{\rm{cv}}^o$, it is nearly impossible for such $\mathbf{\hat{x}}^p$ to be the OMP-CV output. If $\delta_d < 0.1$, e.g., $\left[ 1-\Phi(\lambda) \right]$ would be less than $0.5 \%$ with $\alpha^p = 1$, and further drops to $0.0063 \%$ with $\alpha^p=2$. The probability $\left[ 1-\Phi(\lambda) \right]$ is shown in Figure \ref{fig cv success probability} with $m_{\rm{cv}} = 48$ and variation of $\alpha^p$ and $\delta_d$. As a result, the OMP-CV output recovers all indices of the support set with overwhelming probability. Further, if the OMP-CV output recovers all indices of the support set, its recovery error can be bounded by $C_1 \varepsilon_{\rm{g}}^o$ with high probability showing that the recovery error of OMP-CV output is very close to that of the oracle output.

\begin{remark}
(The rationality of assuming the oracle output recovers all indices in the support set, i.e. $T \subset T^o$) Assume otherwise $T \backslash T^o \neq \emptyset$. Since in OMP-CV the support set of the current iteration contains all indices of the support sets of its previous iterations,

\begin{itemize}
\item if there exists a recovered signal $\mathbf{\hat{x}}^p$ such that $(T^p \backslash T^o) \cap T \neq \emptyset$, it is nearly impossible that $\epsilon_{\rm{cv}}^o < \epsilon_{\rm{cv}}^p$ due to (\ref{not properly recovered signal}). Therefore $\mathbf{\hat{x}}^o$ is nearly impossible to be the oracle output, which contradicts with statements of Theorem \ref{OMP-CV theorem 1}.
\item if $\mathbf{\hat{x}}^o$ has the maximum number of indices of the support $T$ among all recovered signals, then other indices of $T$ are not incorporated in $d$ times of iteration. In this case, $\mathbf{A}_{T \backslash T^o} \mathbf{\hat{x}}_{T \backslash T^o}$ acts as the same role as noise $\mathbf{a}_{\rm{n}} \sigma_{\rm{n}}$. Therefore, they can be treated as noise and similar analysis can be conducted by changing $\sigma_{\rm{n}}^2$ to $(\sigma_{\rm{n}}^2 + \Vert \mathbf{\hat{x}}_{T \backslash T^o}\Vert_2^2)$ and properly modifying the footnote of the related RICs.
\end{itemize}
\end{remark}

\begin{remark} (Parameter details and setting)
Parameter details of Theorem \ref{OMP-CV theorem 1} are: 
\begin{equation*}
\begin{split}
\rm{g}(\alpha^p) = & \frac{\beta_1 (\alpha^p)^2 + \beta_2}{\beta_1 (\alpha^p)^2 + \beta_2 + \max[ (\alpha^p)^2 - \beta_3 \alpha^p - \beta_4, 0]^2} 
\approx \frac{\beta_1}{(\alpha^p)^2 + \beta_1}, \\
C_1 = & 2 C_0 + 1 + 2 \sqrt{C_0^2 + C_0}, \\
C_0 \leq & \beta_5 \frac{\lambda_0^2}{m_{\rm{cv}} - 2 \lambda_0^2},
\end{split}
\end{equation*}
where betas are decided by RICs \cite{candes2005decoding} of the sensing matrix. $\beta_1$ is far larger than $\beta_2$, $\beta_3$, and $\beta_4$. E.g., if $\delta_d <0.1$, the values of betas are: $\beta_1 = 2.08$, $\beta_2 = \beta_3 = 0.03$, $\beta_4 = 0.02$, and $\beta_5 = 0.0376$. 

For parameter setting, $m_{\rm{cv}}$ does not need to be a very large number to attain a promising performance. As $C_0$ is proportional to $1/(m_{\rm{cv}} - 2 \lambda_0^2)$, the value of $C_0$ becomes extremely small when $m_{\rm{cv}}$ becomes a little greater than $2\lambda_0^2$; meanwhile, $C_0$ does not decay much when $m_{\rm{cv}}$ largely exceeds $2 \lambda_0^2$. Therefore $m_{\rm{cv}}$ should be properly chosen to be a little greater than $2\lambda_0^2$. As for $\lambda_0$, the probability $\{1 - (d-k) [1 - \Phi(\lambda_0)]\}$ increase significantly with the increase of $\lambda_0$. For example, it decays by approximately 100 times with $\lambda_0$ increases by 1. For example, when $d-k=100$, by setting $\lambda_0 = 4$, we attain the probability $1 - (d-k) (1 - \Phi(\lambda_0)) = 99.7 \%$. When $d-k$ increase to 10000, we can attain the similar probability of $99.4 \%$ by merely increase $\lambda_0$ by $1$ to be $5$. 

If, e.g., $(d-k) = 100$, setting $m_{\rm{cv}} = 48$ and $\lambda_0 = 4$, produces a numerical form of \eqref{properly recovered signal}: with probability $99.7 \%$ we have 
\begin{equation}
\varepsilon_{\rm{g}}^p \leq 1.47 {\varepsilon_{\rm{g}}^o}.
\end{equation}
\end{remark}

Apart from Theorem \ref{OMP-CV theorem 1}, we also note that the recovery error of OMP-CV output can be estimated in its CV residual in practical cases. This is achieved via a direct application of Theorem \ref{cv theorem 1}:
\begin{remark}(OMP-CV output estimation) 
Let $\mathbf{\hat{x}}^p$ be the OMP-CV output with CV residual $\epsilon_{\rm{cv}}^p$ and assume there are enough measurements for cross validation. It holds with probability $\rm{erf} (\frac{\lambda}{\sqrt{2}})$ that,

\begin{equation}\label{OMP-CV e16}
h(\lambda,+)\epsilon_{\rm{cv}}^p - \sigma_{\rm{n}}^2 \leq \varepsilon_{\rm{x}}^p \leq h(\lambda,-)\epsilon_{\rm{cv}}^p -\sigma_{\rm{n}}^2,
\end{equation}
where $h(\lambda,\pm)$ is defined in Theorem \ref{cv theorem 1}. In practical cases, since $\frac{1}{m}\Vert \mathbf{y}-\mathbf{A}\hat{\mathbf{x}} \Vert_2^2$ is an minimum variance unbiased (MVU) estimator for $\sigma_{\rm{n}}^2$ and $\hat{\mathbf{x}}^p$ is our best result in estimating $\mathbf{x}$, $\sigma_{\rm{n}}^2$ can be approximated by $\hat{\sigma}_{\rm{n}}^2=\frac{1}{m}\Vert \mathbf{y}-\mathbf{A}\hat{\mathbf{x}}^p \Vert_2^2$. Therefore, in practical cases, we could approximate (\ref{OMP-CV e16}) using
\begin{equation}
h(\lambda,+)\epsilon_{\rm{cv}}^p - \hat{\sigma}^2_{\rm{n}} \leq \varepsilon_{\rm{x}}^p \leq h(\lambda,-)\epsilon_{\rm{cv}}^p -\hat{\sigma}^2_{\rm{n}},
\end{equation}
where $\hat{\sigma}_{\rm{n}}^2=\frac{1}{m}\Vert \mathbf{y}-\mathbf{A}\hat{\mathbf{x}}^p \Vert_2^2$.
\end{remark} 
In OMP-CV, therefore, the estimation of the recovery error of the OMP-CV output is available to us in practical cases. If the recovery performance is not as good as expected, some measures could be conducted to improve the recovery performance, e.g., adding more measurements. 
\section{Besides Gaussian: Cross Validation in Other Sensing Matrices}
This section provides a brief discussion of CV performance in the scenario where a general random sensing matrix is used instead of the Gaussian matrix. In the discussion, the matrix setting is similar as previous sections; elements of both reconstruction matrix and CV matrix obey a probability distribution with mean $0$ and variance $1/m$. However, substituting the condition of Gaussian distribution, we now consider a more general one: for elements in the reconstruction matrix $\mathbf{A}_{ij}$ (as well as in the CV matrix), it is given that ${\rm{Var}}(\mathbf{A}_{ij}^2)=\gamma$. Under such condition, imitating Lemma \ref{cv lemma 1} and Lemma \ref{cv lemma 2}, the probability distribution of $\epsilon_{\rm{cv}}$ and $\Delta \epsilon_{\rm{cv}}$ could be given as 

\begin{lemma} \label{gen cv lemma 1}
Let $\mathbf{\hat{x}}$ be the recovered signal and $\Delta x_j$ be the $j$-th term of $\Delta \mathbf{x} = \mathbf{x} - \mathbf{\hat{x}}$. Assuming there are enough measurements for cross validation, one has
\begin{equation}
\epsilon_{\rm{cv}} = \Vert \mathbf{y}_{\rm{cv}} - \mathbf{\Phi}_{\rm{cv}} \mathbf{\hat{x}} \Vert _2^2 \sim N (\mu, \sigma^2),
\end{equation}
where $\mu = \frac{m_{\rm{cv}}}{m} ( \varepsilon_{\rm{x}}+\sigma_{\rm{n}}^2 )$ and $\sigma^2 = \frac{2 m_{\rm{cv}}}{m^2} \left[ (\varepsilon_{\rm{x}}+\sigma_{\rm{n}}^2)^2 + (\frac{m^2}{2}\gamma - 1) \displaystyle\sum _{j=1}^N \Delta x_j^4\right]$.
\end{lemma}

\begin{lemma} \label{gen cv lemma 2}
Let $\mathbf{\hat{x}}^p$ and $\mathbf{\hat{x}}^q$ be two recovered signals and $\Delta x_j^p$ be the $j$-th term of $\Delta \mathbf{x}^p = \mathbf{x} - \mathbf{\hat{x}}^p$. Assuming there are enough measurements for cross validation, one has
\begin{equation}
\Delta \epsilon_{\rm{cv}}^l = \epsilon_{\rm{cv}}^p - \epsilon_{\rm{cv}}^{q} \sim N(\mu, \sigma^2),
\end{equation}
where $\mu = \frac{m_{\rm{cv}}}{m} ( \varepsilon_{\rm{g}}^p - \varepsilon_{\rm{g}}^q )$, $\sigma^2 =\frac{2m_{\rm{cv}}}{m^2} [ (\varepsilon_{\rm{g}}^p)^2 + (\varepsilon_{\rm{g}}^q)^2 - 2\rho_{\rm{g}}^2 \varepsilon_{\rm{g}}^p \varepsilon_{\rm{g}}^q + (\frac{m^2}{2}\gamma - 1) \sum_{j=1}^{N} ((\Delta x_j^p)^2 - (\Delta x_j^q)^2 )^2 ]$, and  
\begin{equation}\rho_{\rm{g}} \triangleq \frac{\langle\Delta \mathbf{x}_{\rm{g}}^p, \Delta \mathbf{x}_{\rm{g}}^q\rangle}{\Vert \Delta \mathbf{x}_{\rm{g}}^p\Vert_2 \Vert\Delta \mathbf{x}_{\rm{g}}^q\Vert_2}
\end{equation}
denotes the correlation coefficient of the two generalized recovery error signals.
\end{lemma}

Compared to their Gaussian version (Lemma \ref{cv lemma 1} and Lemma \ref{cv lemma 2}), in Lemma \ref{gen cv lemma 1} and Lemma \ref{gen cv lemma 2} the mean of $\epsilon_{\rm{cv}}$ and $\Delta \epsilon_{\rm{cv}}$ stay unchanged while there is an extra term in their variance, which is the product of $(\frac{m^2}{2}\gamma-1)$ and a non-negative term. 

We would like to point out that, when the mean of $\epsilon_{\rm{cv}}$ and $\Delta \epsilon_{\rm{cv}}$ stay unchanged, the CV performance, including CV estimation, CV comparison, and definitely its use in OMP-CV, will be better if the variance decrease. In the Gaussian case, the extra terms (both in Lemma \ref{gen cv lemma 1} and Lemma \ref{gen cv lemma 2}) equal to zero as $\gamma=2/m^2$. Therefore, for any other distribution, if $(\frac{m^2}{2}\gamma-1)$ is negative, the CV performance will be no worse than that in the Gaussian case. For example, the Rademacher sensing matrix, with $\gamma=0$ and $(\frac{m^2}{2}\gamma-1)=-1$, obtains a better CV performance than the Gaussian matrix.

\section{Numerical Simulation \label{section Simulation}}
This section gives simulation results concerning our theoretical analyses. Subsection \ref{simulation 1} validates Lemma \ref{cv lemma 1} and Lemma \ref{cv lemma 2} respectively, upon which our general CV techniques were developed. Subsection \ref{simulation 2} simulates the numerical example mentioned in Section \ref{section OMP-CV}, providing evidence supporting the rationality of Theorem \ref{OMP-CV theorem 1}. In Subsection \ref{simulation 3}, the behaviors of OMP-CV algorithm are studied with variation of different parameters. Throughout this section, Gaussian signals are used where non-zero entries are generated following the standard Gaussian distribution.

\subsection{Validation for Lemma \ref{cv lemma 1} and Lemma \ref{cv lemma 2} \label{simulation 1}} 
As previously noted, CLT is used to approximate the probability distribution of $\epsilon_{\rm{cv}}$ in Lemma \ref{cv lemma 1} and $\Delta \epsilon_{\rm{cv}}$ in Lemma \ref{cv lemma 2}. The simulations
in this section attempt to validate these approximations. In both simulations, parameters are set as $N=512$, $m = 96$, $m_{\rm{cv}} = 48$ and $k=50$. With recovered signals ($\mathbf{\hat{x}}$ for
Lemma \ref{cv lemma 1}; $\mathbf{\hat{x}}^p$ and $\mathbf{\hat{x}}^q$ for Lemma \ref{cv lemma 2}) fixed, the random CV matrix along with its noise is realized 1E5 times and the probability distributions of random variables ($\epsilon_{\rm{cv}}$ for Lemma \ref{cv lemma 1}; $\Delta \epsilon_{\rm{cv}}$ for Lemma \ref{cv lemma 2}) are calculated. The experiment results, shown in Figure \ref{fig s1}, indicate that the simulation results agree well with the theoretical prediction. This validates our approximation and supports both lemmas.

\begin{figure}[h]
\begin{center}
\subfigure[]{\includegraphics[scale=0.45]{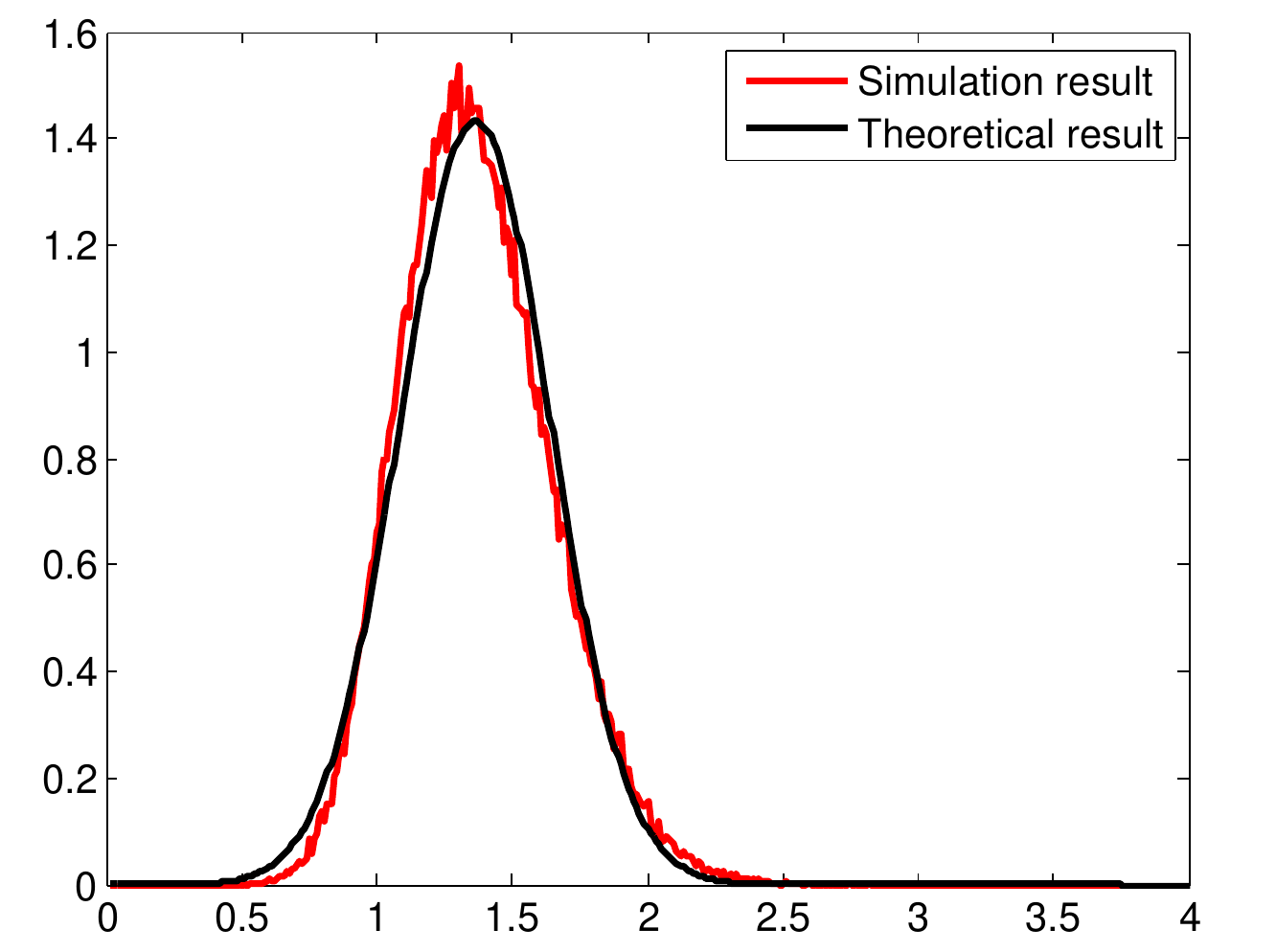}} ~~~~
\subfigure[]{\includegraphics[scale=0.45]{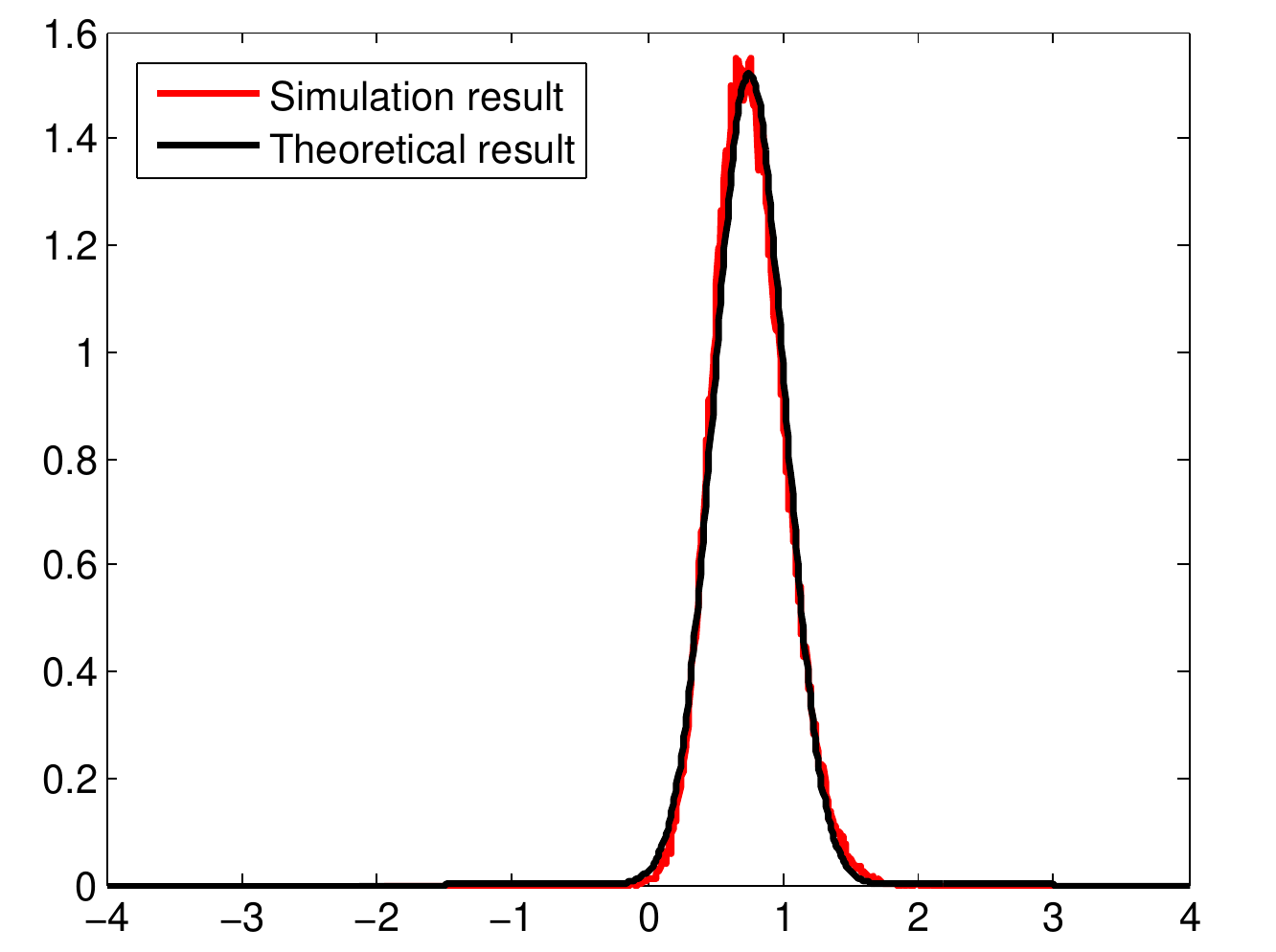}}
\caption{Validation for Lemma \ref{cv lemma 1} and Lemma \ref{cv lemma 2}. Figure \ref{fig s1}(a) validates Lemma \ref{cv lemma 1} by plotting the simulation result of probability distribution of $\epsilon_{\rm{cv}}$, while Figure \ref{fig s1}(b) validates Lemma \ref{cv lemma 2} by plotting that of $\Delta \epsilon_{\rm{cv}}$ (red curve). The simulation results agree well with the theoretical results, which are plotted in black as reference. The Kullback-Leibler divergences of the theoretical results from the simulation results are $0.0152$ and $0.0093$ for Fig. \ref{fig s1}(a) and Fig. \ref{fig s1}(b) respectively. \label{fig s1}}
\end{center}
\end{figure}

\subsection{Validation for Theorem \ref{OMP-CV theorem 1} \label{simulation 2}}
\begin{figure}[h]
\begin{center}
\subfigure[$\sigma_{\rm{n}} = 0.05$]{\includegraphics[scale=0.4]{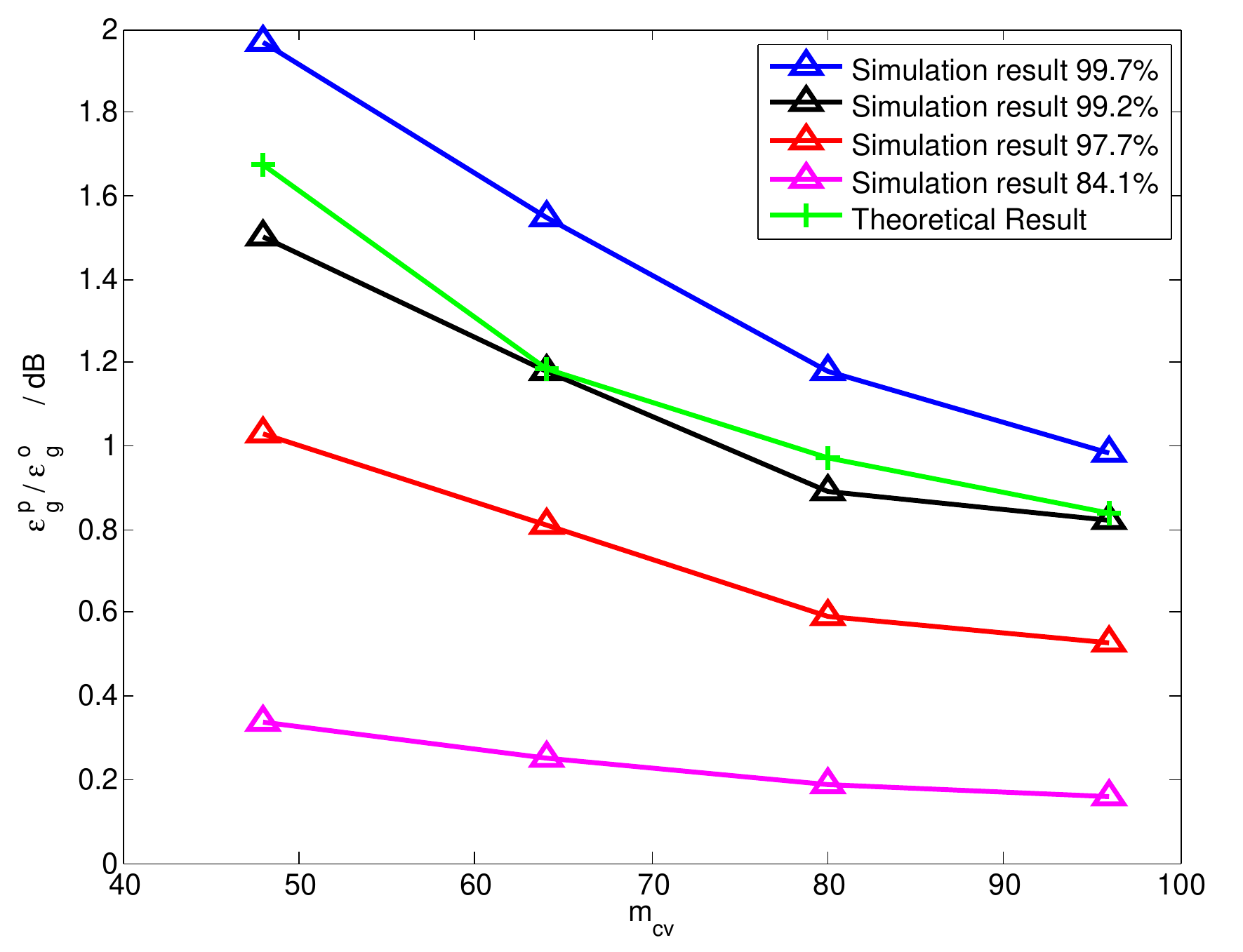}}~~~~
\subfigure[$\sigma_{\rm{n}} = 0.1$]{\includegraphics[scale=0.4]{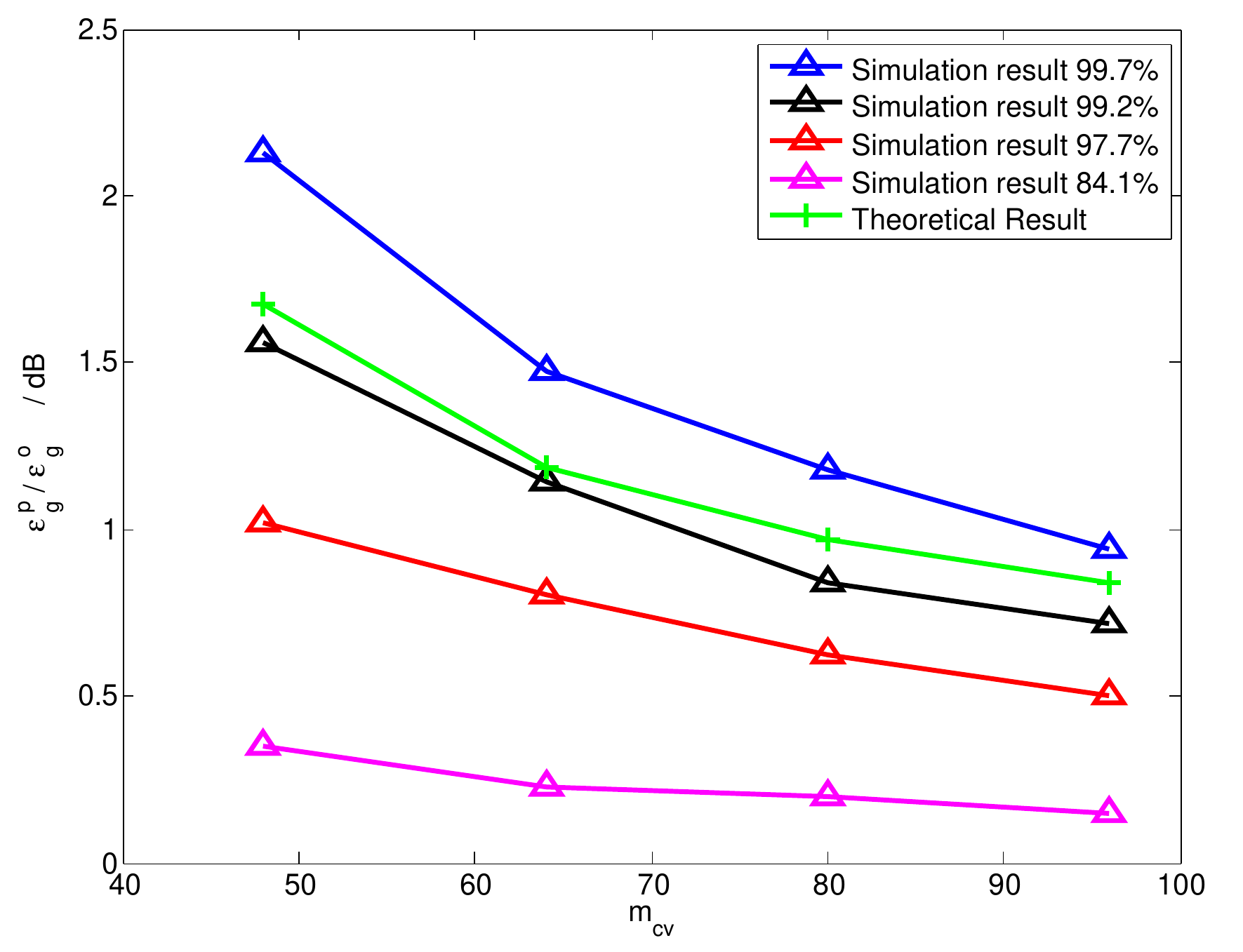}}
\caption{Theoretical results and experimental results of $\varepsilon_{\rm{g}}^p / \varepsilon_{\rm{g}}^o$ under different $m_{\rm{cv}}$ and $\sigma_{\rm{n}}$. Parameters are  $N = 1000$, $m=400$, $k=50$, $d-k = 100$, and $m_{\rm{cv}}$ varies from $48$ to $96$ with step $16$, where with high probability it holds that $\varepsilon_{\rm{g}}^p \leq C_1 {\varepsilon_{\rm{g}}^o}$ according to Theorem \ref{OMP-CV theorem 1}. The theoretical value of $C_1$ that holds with probability $99.7\%$ is plotted in green line (Theoretical result), while the empirical values of $C_1$ that hold with $99.7\%$, $99.2\%$, $97.7\%$, $84.1 \%$ are plotted in blue, black, red, and magenta line respectively (Experimental result). We could see that the experimental value decreases with the increase of $m_{\rm{cv}}$, does not vary when noise level changes, and agrees well with the theoretical result.\label{fig s2}}
\end{center}
\end{figure}

This experiment simulates recovery performance of OMP-CV and compare it with Theorem \ref{OMP-CV theorem 1}. Parameters are set as $N = 1000$, $m=400$, $k=50$, $d-k = 100$, and $m_{\rm{cv}}$ varies from $48$ to $96$ with step $16$. For each parameter setting, the experiment repeated for 5000 times. According to Theorem \ref{OMP-CV theorem 1}, it holds with probability $\left\{ 1 - (d-k) [1 - \Phi(\lambda_0)] \right\}$ that 
\begin{equation}
\varepsilon_{\rm{g}}^p \leq C_1 {\varepsilon_{\rm{g}}^o},
\end{equation}
provided that oracle output and CV output recovered all indices in the support set. The theoretical bound $C_1$ is calculated using 
\begin{gather}
C_1 = 2 C_0 + 1 + 2 \sqrt{C_0^2 + C_0}, \\
C_0 \leq \beta_5 \frac{\lambda_0^2}{m_{\rm{cv}} - 2 \lambda_0^2},
\end{gather}
where $\lambda_0=4$ so that $\left\{ 1 - (d-k) [1 - \Phi(\lambda_0)] \right\} = 99.7\%$ and $\beta_5=0.0376$ under the assumption that $\delta_d<0.1$.

From the result we can see the simulation values decrease with the increase of $m_{\rm{cv}}$ and do not vary when noise level changes, which agree with the theoretical analysis. Furthermore, we could see that with overwhelming probability the simulation error ratio $\varepsilon_{\rm{g}}^p / \varepsilon_{\rm{g}}^o$ is smaller than the theoretical bound and in most cases it remains pretty low. $84.1\%$ of the recovered signal, e.g., have a error ratio lower than $0.4 dB$. More precisely, the percentage that the error ratio does not exceeds the theoretical bound is roughly $99.5\%$, which is approximately the theoretical probability in our theorem, $99.7 \%$. The actual proportion is a little bit smaller than the theoretical probability because the RIP assumption and the CLT approximation may slightly vary with the practical situation.
\subsection{Simulation for OMP-CV Algorithm \label{simulation 3}}
This subsection investigates the behaviors of OMP-CV with variation of different parameters. The performance of OMP, referred to as \emph{OMP-residual}, is also given accordingly as reference. The stopping criteria of OMP-residual is based on its residual, i.e. terminating the iteration as long as 
\begin{equation}
\Vert \mathbf{y} - \mathbf{A} \mathbf{\hat{x}} \Vert_2^2 < \sigma_n^2.
\end{equation}
We would like to note that the prior knowledge of noise level is required in OMP-residual but not in OMP-CV. For each parameter setting, algorithms are repeated 1000 times. 
\subsubsection{Number of CV Measurements}
This experiment investigates the effect of additional cross validation measurements. Parameters are set as $N=1000$, $k=50$, and $\sigma_n^2 =0.1$. The number of reconstruction measurement $m$ is fixed to be 360 and the number of CV measurements $m_{\rm{cv}}$ varies from $10$ to $80$ with step $10$. The experiment result is shown in Figure \ref{fig s3ss1}, from which we can see with the increase of $m_{\rm{cv}}$ improves the recovery performance significantly. We can also see the improvement of the recovery performance becomes slow when $m_{\rm{cv}}$ exceeds $60$. Thus, it is not necessary for $m_{\rm{cv}}$ to be very large to obtain a satisfying recovery performance.

\begin{figure}[h]
\begin{center}
\includegraphics[scale=0.6]{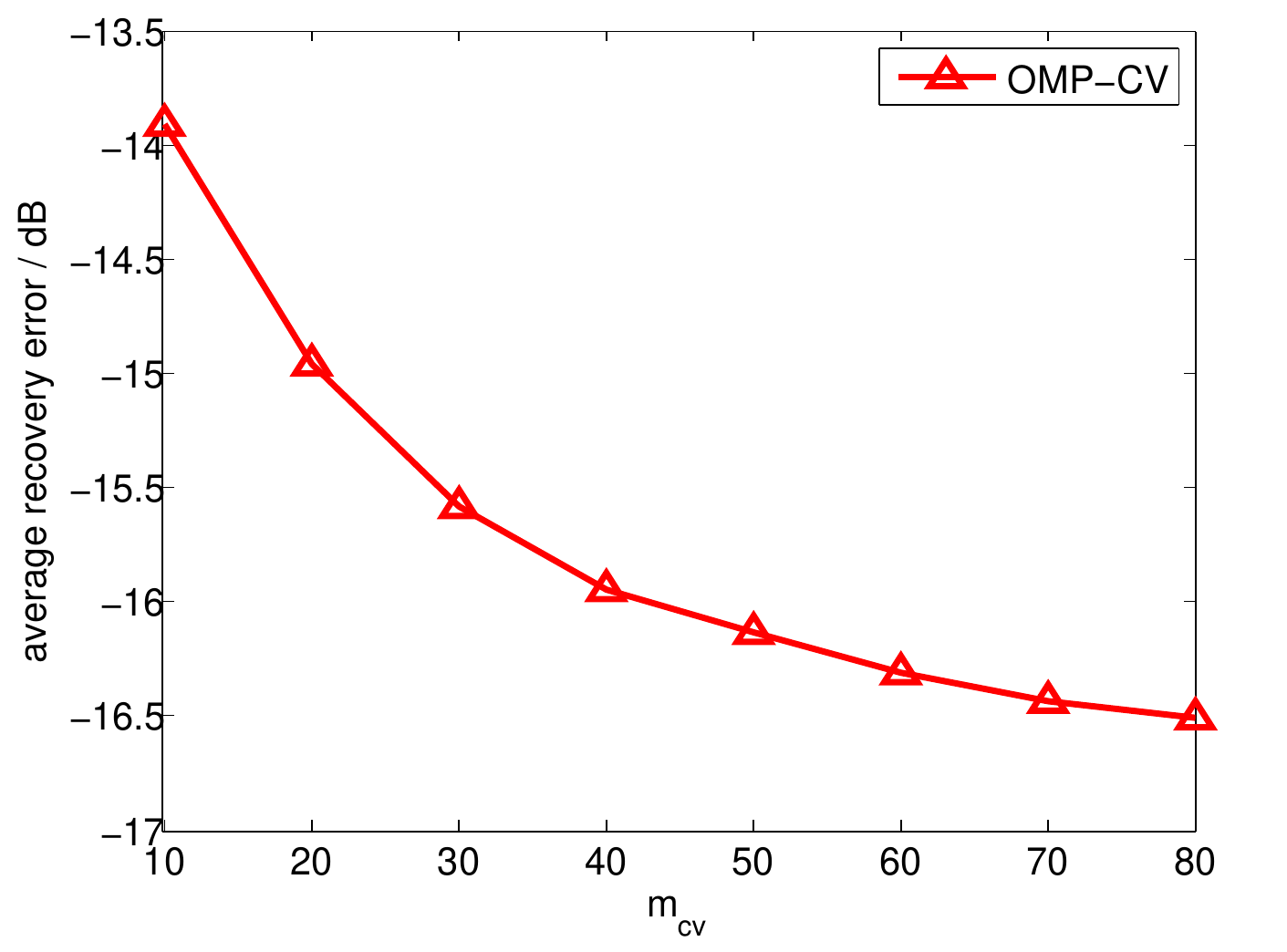}
\caption{The effect of additional CV measurements. Parameters are $N=1000$, $k=50$, $m=360$, and $\sigma_n^2 =0.1$. $m_{\rm{cv}}$ varies from $10$ to $80$ with step $10$. The experiment result shows that the increase of $m_{\rm{cv}}$ improves the recovery performance significantly and the performance improvement becomes slow when $m_{\rm{cv}}$ exceeds $60$. Thus, it is not necessary for $m_{\rm{cv}}$ to be very large to obtain a satisfying recovery performance. \label{fig s3ss1}}
\end{center}
\end{figure}
\subsubsection{Trade Off between $m$ and $m_{\rm{cv}}$}
Given a fixed total number of measurements $M$, there is a tradeoff between $m$ and $m_{\rm{cv}}$ \cite{boufounos2007sparse}. On one hand, increasing $m$ will reduce the reconstruction error. On the other hand, increasing $m_{\rm{cv}}$ will improve the CV estimation, and thus make the OMP-CV output closer to the oracle output. This simulation empirically investigates the recovery performance of OMP-CV as $m_{\rm{cv}}$ varies.

In this experiment, we set $N=1000$, $k=50$, $\sigma_n^2 =0.1$, and $M=400$. $m_{\rm{cv}}$ varies from $100$ to $10$ with step $-10$ and we have $m = M - m_{\rm{cv}}$. For OMP-CV, $m$ measurements are used for reconstruction, while $m_{\rm{cv}}$ measurements are used for cross validation. For OMP-residual, $m$ measurements are used for reconstruction and the termination is based on residual with the accurate noise level given. In addition, the recovery performance, where all $M$ measurements are used for reconstruction using OMP-residual, is given as reference. We average 1000 repetitions for experiments of each parameter setting.

The experiment result, plotted in Figure \ref{fig s3ss2}, shows the best performance of OMP-CV lies in the region where $m_{\rm{cv}}$ is neither too small nor too large. OMP-CV outperforms OMP-residual except when $m_{\rm{cv}}$ is very small, indicating that CV-based termination is better than residual-based termination, even if the latter uses an exact noise level. Additionally, note that with the same number of measurements at hand ($M$ measurements for both OMP-CV and OMP-residual), OMP-CV can achieve recovery performance similar to OMP-residual with parameters appropriately set, even though prior knowledge is required for OMP-residual. In this sense, OMP-CV outperforms OMP-residual.
\begin{figure}[h]
\begin{center}
\includegraphics[scale=0.6]{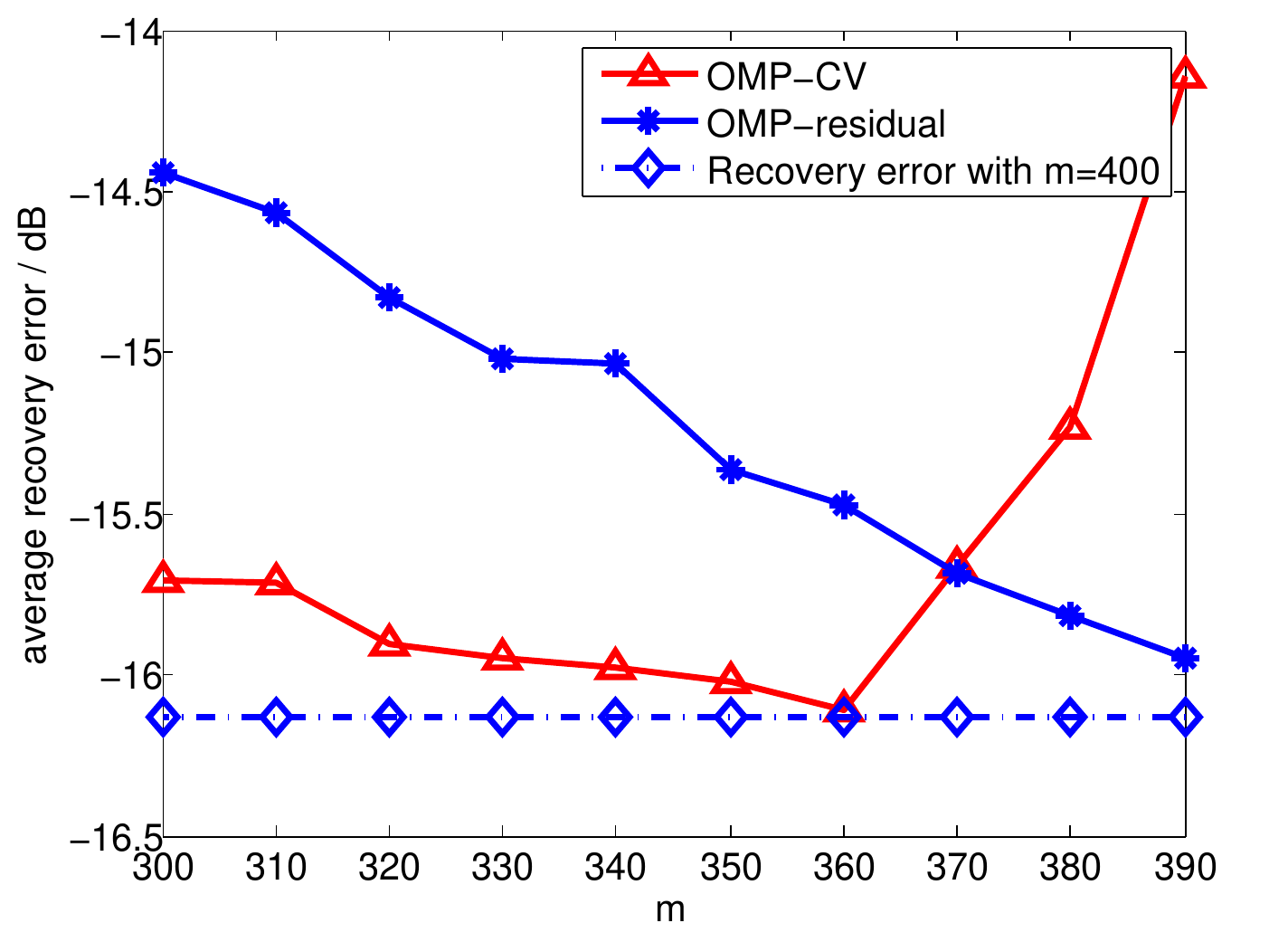}
\caption{The trade Off between $m$ and $m_{\rm{cv}}$. Parameters are  fixed as $N=1000$, $k=50$, $\sigma_n^2 =0.1$, and   $M=400$. $m_{\rm{cv}}$ varies from $100$ to $10$ with step $-10$ while $m = M - m_{\rm{cv}}$. OMP-CV outperforms OMP except when $m_{\rm{cv}}$ is too small. In addition, using same number of measurements, OMP-CV has recovery performance similar to OMP-residual (recovery error with $m=400$) with parameters appropriately set even though the prior knowledge is required for OMP-residual. \label{fig s3ss2}}
\end{center}
\end{figure}
\subsubsection{The Performance of OMP-CV under Different Noise Level}
This experiment investigates the recovery performance of OMP-CV under different noise level. In addition to the reference performance of OMP-residual, the performance of OMP that is conducted with no prior knowledge, referred to as \emph{OMP-without-prior-knowledge}, is given as baseline. For OMP-without-prior-knowledge, the stopping criteria \cite{donohosparselab} is to terminate the iteration either when iteration times exceeds $M$ or when
\begin{equation}
\Vert \mathbf{y} - \mathbf{A} \mathbf{\hat{x}} \Vert_2 < 10^{-5} \Vert \mathbf{y} \Vert_2. 
\end{equation}
The number of total measurements $M$ is fixed to be $400$. For OMP-CV, $m=352$ measurements are used for reconstruction while $m_{\rm{cv}}=48$ measurements for cross validation. For OMP-residual and OMP-without-prior-knowledge, all $M$ measurements are used for reconstruction. The noise level $\sigma_n^2$ varies from $0.02$ to $0.2$ with step $0.02$. Other parameters are $N=1000$, and $k=50$. The experiment result is shown in Figure \ref{fig s3ss3}. 

From Figure \ref{fig s3ss3} we can see that OMP-CV and OMP-residual have similar performance under different noise level even when the prior knowledge is available for OMP-residual. We also observe that the performance of OMP largely deteriorates when information of noise level is not available (OMP-without-prior-knowledge). In this sense, OMP-CV outperforms OMP-residual.  

\begin{figure}[h]
\begin{center}
\includegraphics[scale=0.6]{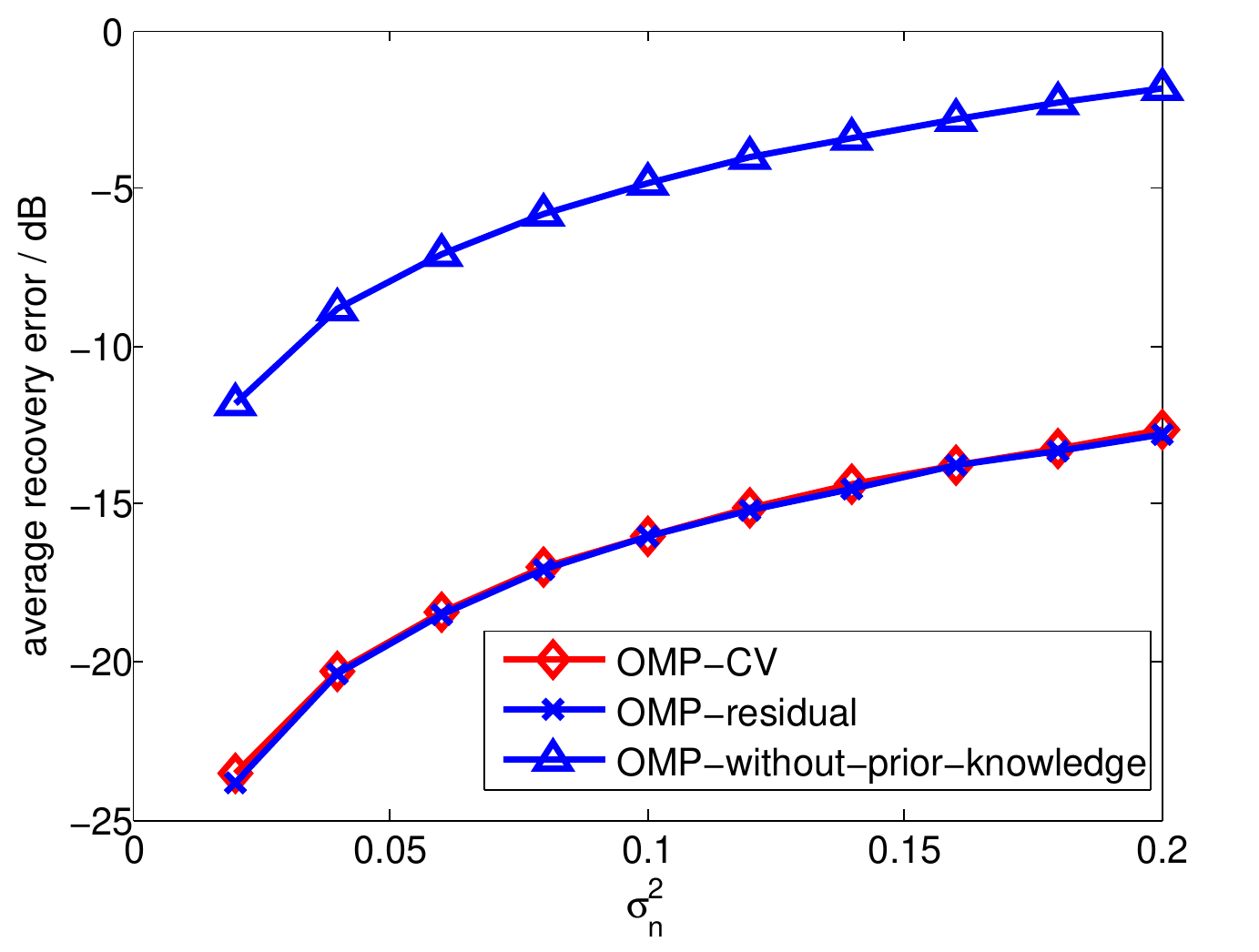}
\caption{The performance of OMP-CV, OMP-residual, and OMP-without-prior-knowledge under different noise level. Parameters are $N=1000$, $k=50$, $M=500$, and $m_{\rm{cv}} = 48$. $\sigma_n^2$ varies from $0.02$ to $0.2$ with step $0.02$. OMP-CV and OMP-residual have similar performance under different noise level even when the prior knowledge is available for OMP-residual. In addition, the performance of OMP-without-prior-knowledge is very bad where no prior knowledge is given to decide the stopping condition of OMP. \label{fig s3ss3}}
\end{center}
\end{figure}

\section{Conclusion \label{section conclusion}}
This paper presents a theoretical study of CV in compressive sensing, providing analysis of general CV problems as well as analysis of the OMP-CV algorithm. As a highly practical algorithm, OMP-CV could reconstruct the signal without prior knowledge such as sparsity or noise level; its performance is supported in this paper both theoretically and empirically. Additionally, our results on general CV problems could also be used to understand other CS-based reconstruction algorithms.  

CV sacrifices a small amount of measurements to estimate the reconstruction error. In a nutshell, this technique makes it possible for greedy algorithms to reconstruct the signal without prior knowledge like the sparsity or noise level. In future work, we would like to extend our analysis on CV by studying the use of CV in other greedy sparse recovery algorithms.
\section{Acknowledgement}
The auther Jinye Zhang would like to thank to Wensi You for her encouragement, without which this work could hardly be completed. He would also like to thank Xinyue Shen for her helpful suggestions, which made the research process much smoother. 
\section*{Appendix}
\appendix
\section{proof of Lemma \ref{RIP lemma 4}}
 \label{app RIP lemma 4}
\begin{proof}
Lemma \ref{RIP lemma 1} shows that, 
\begin{equation}\label{App e1}
\Vert \mathbf{A}_S^\dagger \mathbf{A}_T \mathbf{x}_T\Vert_2 = \Vert (\mathbf{A}_S'\mathbf{A}_S)^{-1}\mathbf{A}_S' \mathbf{A}_T \mathbf{x}_T\Vert_2 \leq \frac{1}{1-\delta_{\vert S \vert}} \Vert\mathbf{A}_S' \mathbf{A}_T \mathbf{x}_T\Vert_2.
\end{equation}
Furthermore, Lemma \ref{RIP lemma 3} implies that,
\begin{equation}\label{App e2}
\Vert\mathbf{A}_S' \mathbf{A}_T 
\mathbf{x}_T\Vert_2 \leq \delta_{\vert S \vert + \vert T \vert} \Vert \mathbf{x}_T\Vert_2.
\end{equation}
Substitute (\ref{App e2}) into (\ref{App e1}) to complete the proof.
\end{proof}
\section{proof of Lemma \ref{cv lemma 1}}
\label{app cv lemma 1}
\begin{proof}
Write the CV residual as, 
\begin{equation}
\epsilon_{\rm{cv}} = \Vert \mathbf{y}_{\rm{cv}} - \mathbf{A}_{\rm{cv}} \mathbf{\hat{x}} \Vert _2^2 = \Vert \mathbf{A}_{\rm{cv}} (\mathbf{x}-\mathbf{\hat{x}}) + \mathbf{n}_{\rm{cv}} \Vert _2^2 = \displaystyle \sum_{i=1}^{m_{\rm{cv}}} \left( \sum _ {j=1} ^N a_{{\rm{cv}}_{ij}} \Delta x_j + n_{{\rm{cv}}_i} \right)^2,
\end{equation}
where $\Delta x_i$ is the $i$-th element of $\Delta \mathbf{x} = \mathbf{x}-\mathbf{\hat{x}}$ and $a_{{\rm{cv}}_{ij}}$ is the element in the $i$-th row and $j$-th column of $\mathbf{A}_{\rm{cv}}$. 

Define that $r_i=\sum _ {j=1} ^N a_{{\rm{cv}}_{ij}} \Delta x_j + n_{{\rm{cv}}_i}$, one has $\epsilon_{\rm{cv}} = \sum _{i=1}^{m_{\rm{cv}}} r_i^2$. Calculate the mean and variance of $r_i$ as,
\begin{equation}
{\rm{E}}(r_i) = \sum _ {j=1} ^N \Delta x_j {\rm{E}}(a_{{\rm{cv}}_{ij}}) + {\rm{E}}(n_{{\rm{cv}}_i}) = 0,
\end{equation}
and
\begin{equation}\label{AppB e1}
{\rm{Var}}(r_i) =\sum _ {j=1} ^N \Delta x_j^2 {\rm{E}}(a_{{\rm{cv}}_{ij}}^2) + {\rm{E}}(n_{{\rm{cv}}_i}^2) = \frac{1}{m}\left(\sum _ {j=1} ^N \Delta x_j^2 + \sigma_{\rm{n}}^2 \right).
\end{equation}
The first equation of (\ref{AppB e1}) holds because $a_{{\rm{cv}}_{ij}}$ and $n_{{\rm{cv}}_i}$ are mutually independent. Recall that $\varepsilon_{\rm{x}} = \sum _ {j=1} ^N \Delta x_j^2$, we further obtain ${\rm{Var}}(r_i) = (\varepsilon_{\rm{x}} + \sigma_{\rm{n}}^2) / m$.
 
Being linear combination of Gaussian variables, $r_i$ is Gaussian distributed and therefore, 
\begin{equation}
r_i^2 \sim \frac{1}{m} (\varepsilon_{\rm{x}} + \sigma_{\rm{n}}^2) \chi_1^2,
\end{equation}
which has mean $(\varepsilon_{\rm{x}} + \sigma_{\rm{n}}^2) / m$ and variance $2(\varepsilon_{\rm{x}} + \sigma_{\rm{n}}^2) ^2 / m^2$. Under the assumption that $m_{\rm{cv}}$ is large, the Central Limit Theorem implies that
\begin{equation}
\epsilon_{\rm{cv}} = \sum_{i=1}^ {m_{\rm{cv}}} r_i^2 \sim \mathcal{N}\left(\frac{m_{\rm{cv}}}{m}(\varepsilon_{\rm{x}} + \sigma_{\rm{n}}^2), \frac{2m_{\rm{cv}}}{m^2}(\varepsilon_{\rm{x}} + \sigma_{\rm{n}}^2)^2\right).
\end{equation} 
\end{proof}
\section{proof of Lemma \ref{cv lemma 2}}
\label{app cv lemma 2}
\begin{proof}
Write $\Delta \epsilon_{\rm{cv}}=\epsilon_{\rm{cv}}^p - \epsilon_{\rm{cv}}^q$ as,
\begin{equation}
\begin{split}
\Delta \epsilon_{\rm{cv}} = & \displaystyle \sum_{i=1}^{m_{\rm{cv}}} \left[ \left(\sum_{j=1}^{N} a_{{\rm{cv}}_{ij}} \Delta x_j^p + n_{{\rm{cv}}_i}\right)^2 - \left(\sum_{j=1}^{N} a_{{\rm{cv}}_{ij}} \Delta x_j^q + n_{{\rm{cv}}_i}\right)^2 \right]\\
=& \sum_{i=1}^{m_{\rm{cv}}} \left[ \left(\sum_{j=1}^{N} a_{{\rm{cv}}_{ij}} (\Delta x_j^p + \Delta x_j^q) \right)\left(\sum_{k=1}^{N}a_{{\rm{cv}}_{ik}}(\Delta x_k^p - \Delta x_k^q)\right) + 2 n_{{\rm{cv}}_i}\left(\sum_{k=1}^{N}a_{{\rm{cv}}_{ik}}(\Delta x_k^p - \Delta x_k^q)\right) \right],
\end{split}
\end{equation}
where $a_{{\rm{cv}}_{ij}}$ is the element in the $i$-th row and $j$-th column of $\mathbf{A}_{\rm{cv}}$. Define that $$r_i = \displaystyle\left[\left(\sum_{j=1}^{N} a_{{\rm{cv}}_{ij}} (\Delta x_j^p + \Delta x_j^q) \right)\left(\sum_{k=1}^{N}a_{{\rm{cv}}_{ik}}(\Delta x_k^p - \Delta x_k^q)\right) + 2 n_{{\rm{cv}}_i}\left(\sum_{k=1}^{N}a_{{\rm{cv}}_{ik}}(\Delta x_k^p - \Delta x_k^q)\right)\right],$$ one has $\Delta \epsilon_{\rm{cv}} = \sum_{i=1}^{m_{\rm{cv}}} r_i$.

We next calculate the mean and variance of $r_i$. $r_i$ is the linear combination of random variables of three kinds: $a_{{\rm{cv}}_{ij}}^2$, $a_{{\rm{cv}}_{ij}}a_{{\rm{cv}}_{ik}}$ and $n_{{\rm{cv}}_i}a_{{\rm{cv}}_{ij}}$, whose means and variances are,
\begin{align}
{\rm{E}} (a_{{\rm{cv}}_{ij}}^2) &= 1 / m ~~&~~ {\rm{Var}} (a_{{\rm{cv}}_{ij}}^2) &= 2 / m^2,\\
{\rm{E}} (a_{{\rm{cv}}_{ij}}a_{{\rm{cv}}_{ik}}) &= 0 ~~&~~ {\rm{Var}} (a_{{\rm{cv}}_{ij}}a_{{\rm{cv}}_{ik}}) &= 1 / m^2,\\
{\rm{E}} (n_{{\rm{cv}}_i}a_{{\rm{cv}}_{ij}}) &= 0 ~~&~~ {\rm{Var}} (n_{{\rm{cv}}_i}a_{{\rm{cv}}_{ij}}) &= \sigma_{\rm{n}}^2 / m^2.
\end{align}
Also, any two of these three types of random variables are mutually independent, therefore,
\begin{equation}
{\rm{Cov}} (x,y)=0
\end{equation}
where $(x, y) = (a_{{\rm{cv}}_{ij}}^2, a_{{\rm{cv}}_{ij}}a_{{\rm{cv}}_{ik}})$, $(a_{{\rm{cv}}_{ij}}^2, n_{{\rm{cv}}_i}a_{{\rm{cv}}_{ij}})$, or $(a_{{\rm{cv}}_{ij}}a_{{\rm{cv}}_{ik}}, n_{{\rm{cv}}_i}a_{{\rm{cv}}_{ij}})$. 

Calculate the mean of $r_i$ as,
\begin{equation}
{\rm{E}} (r_i) = \frac{1}{m} \displaystyle \sum_{j=1}^{N} (\Delta x_j^p + \Delta x_j^q)(\Delta x_j^p - \Delta x_j^q) = \frac{1}{m}(\varepsilon_{\rm{x}}^p - \varepsilon_{\rm{x}}^q) = \frac{1}{m}(\varepsilon_{\rm{g}}^p - \varepsilon_{\rm{g}}^q).
\end{equation} 

Write the variance of $r_i$ as, 
\begin{equation} \label{App e4}
\begin{split}
{\rm{Var}} (r_i) = {{\rm{Var}}} \left(\displaystyle\sum_{j=1}^{N} a_{{\rm{cv}}_{ij}} (\Delta x_j^p + \Delta x_j^q) \sum_{k=1}^{N}a_{{\rm{cv}}_{ik}}(\Delta x_k^p - \Delta x_k^q)\right) + 4 {\rm{Var}} \left( n_{cv_i}\sum_{k=1}^{N}a_{{\rm{cv}}_{ik}}(\Delta x_k^p - \Delta x_k^q)\right).
\end{split}
\end{equation}

Next calculate both terms of (\ref{App e4}) respectively,  
\begin{equation}
\begin{split}
{\rm{first~term}} &= \textrm{Var}\left(\displaystyle\sum_{j=1}^{N} a_{{\rm{cv}}_{ij}} (\Delta x_j^p + \Delta x_j^q) \sum_{k=1}^{N}a_{{\rm{cv}}_{ik}}(\Delta x_k^p - \Delta x_k^q)\right)\\
&= \frac{2}{m^2}\sum_{j=1}^{N}(\Delta x_j^p+\Delta x_j^q)^2 (\Delta x_j^p-\Delta x_j^q)^2 \\ 
& + \frac{1}{m^2}\sum_{j=1}^{N} \sum_{k=j+1}^N \left((\Delta x_j^p+\Delta x_j^q)(\Delta x_k^p-\Delta x_k^q) + (\Delta x_k^p + \Delta x_k^q)(\Delta x_j^p-\Delta x_j^q)\right)^2\\
&= \frac{2}{m^2}\sum_{j=1}^{N}\left((\Delta x_j^p)^2 - (\Delta x_j^q)^2\right)^2 + \frac{4}{m^2} \sum_{j=1}^{N} \sum_{k=j+1}^N \left(\Delta x_j^p \Delta x_k^p - \Delta x_j^q \Delta x_k^q\right)^2\\
&= \frac{2}{m^2} \left[\left(\sum_{j=1}^{N}(\Delta x_j^p)^2\right)^2 + \left(\sum_{j=1}^{N}(\Delta x_j^q)^2\right)^2 - 2 \left(\sum_{j=1}^{N} \Delta x_j^p \Delta x_j^q\right)^2\right]\\
&=  \frac{2}{m^2}  \left[ (\varepsilon_{\rm{x}}^p)^2 + (\varepsilon_{\rm{x}}^q)^2 - 2\langle \Delta \mathbf{x}^p, \Delta \mathbf{x}^q\rangle ^2 \right];
\end{split}
\end{equation}
\begin{equation}
{\rm{second~term}} = 4 \textrm{Var} \left( n_{cv_i}\sum_{k=1}^{N}A_{cv_{ik}}(\Delta x_k^p - \Delta x_k^q)\right) = \displaystyle \frac{4\sigma_{\rm{n}}^2}{m^2} \sum_{j=1}^{N} (\Delta x_j^p - \Delta x_j^q)^2 = \frac{4\sigma_{\rm{n}}^2}{m^2} \Vert \Delta \mathbf{x}^p - \Delta \mathbf{x}^q \Vert_2^2.
\end{equation}
Therefore, 
\begin{equation}
\begin{split}
\textrm{Var} (R_i) &= \frac{2}{m^2}  \left[ (\varepsilon_{\rm{x}}^p)^2 + (\varepsilon_{\rm{x}}^q)^2 - 2\langle \Delta \mathbf{x}^p, \Delta \mathbf{x}^q\rangle ^2 + 2 \sigma_{\rm{n}}^2 \Vert \Delta \mathbf{x}^p - \Delta \mathbf{x}^q \Vert_2^2 \right] \\
&=\frac{2}{m^2} \left[ (\varepsilon_{\rm{g}}^p)^2 + (\varepsilon_{\rm{g}}^q)^2 - 2\rho_{\rm{g}}^2 \varepsilon_{\rm{g}}^p \varepsilon_{\rm{g}}^q \right].
\end{split}
\end{equation}

Furthermore, under the assumption that $m_{\rm{cv}}$ is large, the Central Limit Theorem implies that,   
\begin{equation}
\Delta \epsilon_{\rm{cv}} = \displaystyle \sum_{i=1} ^ {m_{\rm{cv}}} r_i \sim \mathcal{N}( \mu, \sigma^2),
\end{equation}
where $\mu = \frac{m_{\rm{cv}}}{m}( \varepsilon_{\rm{g}}^p - \varepsilon_{\rm{g}}^q)
$ and $\sigma^2 =\frac{2m_{\rm{cv}}}{m^2} [(\varepsilon_{\rm{g}}^p)^2 + (\varepsilon_{\rm{g}}^q)^2 - 2\rho_{\rm{g}}^2 \varepsilon_{\rm{g}}^p \varepsilon_{\rm{g}}^q]$.
\end{proof}
\section{proof of Theorem \ref{OMP-CV theorem 1}}\label{app OMP-CV theorem 1}
\begin{proof}
The proof of Theorem \ref{OMP-CV theorem 1} consists of three steps. Firstly, the relation between recovered signals generated in two iterations of OMP-CV is studied. Next, general CV techniques are used to calculate the CV comparison success probability of these two recovered signals. Finally, we analyse this probability in both situations of Theorem \ref{OMP-CV theorem 1} to complete the proof.
 
\subsection{Step1: Relation between Two Recovered Signals}
Consider recovered signals generated in the $p$-th and $q$-th iteration, $\mathbf{\hat{x}}^p$ and $\mathbf{\hat{x}}^q$, then $\mathbf{\hat{x}}_{T^p}^p = \mathbf{A}_{T^p}^ \dagger \mathbf{y}$ and $\mathbf{\hat{x}}_{T^q}^q = \mathbf{A}_{T^q}^ \dagger \mathbf{y}$. Assume next that $p<q$. The relation between $\mathbf{\hat{x}}^p$ and $\mathbf{\hat{x}}^q$ can be described as the following lemma.

\begin{lemma} \label{OMP-CV lemma 1}
Let $\mathbf{\hat{x}}^p$ and $\mathbf{\hat{x}}^q$ be recovered signals generated in the $p$-th and $q$-th iteration in OMP-CV and $p < q$. It holds that,
\begin{equation}
\mathbf{\hat{x}}_{T^q}^q = \mathbf{A}_{T^q}^ \dagger \mathbf{y} = 
\left[ \begin{array}{c}
\mathbf{A}_{T^p} ^ \dagger (\mathbf{A}_{(T^q)^c} \mathbf{x}_{(T^q)^c} + \mathbf{n} - \mathbf{A}_{T^{q-p}} \bm{\delta}_{T^{q-p}})\\ 
\bm{\delta}_{T^{q-p}}
\end{array} \right]
+ \mathbf{x}_{T^q},
\end{equation}
where $\bm{\delta}_{T^{q-p}} = (\mathbf{A}_{T^{q-p}}' P_{T^p} \mathbf{A}_{T^{q-p}})^{-1} \mathbf{A}_{T^{q-p}}' P_{T^p} (\mathbf{A}_{(T^q)^c} \mathbf{x}_{(T^q)^c} + \mathbf{n})$ and $T^{q-p}$ denotes the set of indices $T^q \backslash T^p$. \footnote{For simplicity, the elements of $\mathbf{\hat{x}}_{T^q}^q$ is reordered. Similar operations are conducted to the analysis below.}
\end{lemma}
\begin{proof}
The proof of Lemma \ref{OMP-CV lemma 1} is in Appendix \ref{app OMP-CV lemma 1}.
\end{proof}

Write $\Delta \mathbf{x}^p$ and $\Delta \mathbf{x}^q$ according to Lemma \ref{OMP-CV lemma 1},
\begin{gather}
\Delta \mathbf{x}^p = \mathbf{x}_T - \mathbf{\hat{x}}_T^p = \left[ \begin{array}{c}
-\mathbf{A}_{T^p}^\dagger (\mathbf{A}_{(T^q)^c} \mathbf{x}_{(T^q)^c} + \mathbf{n}) \\
\mathbf{x}_{T^{q-p}}\\
\mathbf{x}_{(T^q)^c}
\end{array} \right],\\
\Delta \mathbf{x}^q = \mathbf{x}_T - \mathbf{\hat{x}}_T^q = \left[ \begin{array}{c}
-\mathbf{A}_{T^p}^\dagger (\mathbf{A}_{(T^q)^c} \mathbf{x}_{(T^q)^c} + \mathbf{n} - \mathbf{A}_{T^{q-p}}\bm{\delta}_{T^{q-p}}) \\
- \bm{\delta}_{T^{q-p}}\\
\mathbf{x}_{(T^q)^c}
\end{array} \right],
\end{gather}
where $(T^q)^c$ denotes the index set $T \backslash T^q$. Bounding $\Vert \bm{\delta}_{T^{q-p}} \Vert_2^2$ using the following lemma.
\begin{lemma}\label{OMP-CV lemma 2}
\begin{equation}
\Vert \bm{\delta}_{T^{q-p}}\Vert_2^2 \leq \eta (\Vert \mathbf{x}_{(T^q)^c} \Vert_2^2 + \sigma_{\rm{n}}^2) 
\end{equation}
where $\eta$ is related to RICs of the sensing matrix $\mathbf{A}$. With the hypothesis $\delta_d \leq 0.1$ we have $\eta \leq 0.0127$.
\end{lemma}
\begin{proof}
The proof of Lemma \ref{OMP-CV lemma 2} can be found in Appendix \ref{app OMP-CV lemma 2}.
\end{proof}

\subsection{Step2: Using General CV Techniques}
According to Theorem \ref{cv theorem 2},
\begin{equation}\label{OMP-CV e3}
\frac{1}{\lambda^2} = \frac{2}{m_{\rm{cv}}} \left[1+2(1-\rho_{\rm{g}}^2)\frac{\varepsilon_{\rm{g}}^q \varepsilon_{\rm{g}}^p}{(\varepsilon_{\rm{g}}^q - \varepsilon_{\rm{g}}^p)^2}\right].
\end{equation}
To calculate the CV comparison success probability $\Phi(\lambda)$, we next calculate 
\begin{equation} 
(1-\rho_{\rm{g}}^2)\frac{\varepsilon_{\rm{g}}^q \varepsilon_{\rm{g}}^p}{(\varepsilon_{\rm{g}}^q - \varepsilon_{\rm{g}}^p)^2} = \frac{\varepsilon_{\rm{g}}^q \varepsilon_{\rm{g}}^p - \langle \Delta \mathbf{x}_{\rm{g}}^q, \Delta \mathbf{x}_{\rm{g}}^p \rangle ^2}{(\varepsilon_{\rm{g}}^q - \varepsilon_{\rm{g}}^p)^2}.
\end{equation}

It follows from Lemma \ref{OMP-CV lemma 1} that
\begin{gather}
\label{OMP-CV e8} \varepsilon_{\rm{g}}^p = \sigma_{\rm{n}}^2 + \Vert \mathbf{x}_{(T^q)^c}\Vert_2^2 + \Vert \mathbf{x}_{T^{q-p}} \Vert_2^2 + \Vert \mathbf{A}_{T^p}^\dagger (\mathbf{A}_{(T^p)^c} \mathbf{x}_{(T^p)^c}+ \mathbf{n}) \Vert_2^2, \\
\label{OMP-CV e9}\varepsilon_{\rm{g}}^q = \sigma_{\rm{n}}^2 + \Vert \mathbf{x}_{(T^q)^c}\Vert_2^2 + \Vert \bm{\delta}_{T^{q-p}} \Vert_2^2 + \Vert \mathbf{A}_{T^p}^\dagger [\mathbf{A}_{(T^p)^c}\mathbf{x}_{(T^p)^c} + \mathbf{n} - \mathbf{A}_{T^{q-p}}(\mathbf{x}_{T^{q-p}} + \bm{\delta}_{T^{q-p}})]\Vert_2^2.
\end{gather}
Define,
\begin{align}
\label{OMP-CV e10} \mathbf{a} &\triangleq \mathbf{A}_{T^p}^\dagger (\mathbf{A}_{(T^p)^c} \mathbf{x}_{(T^p)^c} + \mathbf{n}), \\
\label{OMP-CV e11} \mathbf{b} &\triangleq \mathbf{A}_{T^p}^\dagger [\mathbf{A}_{(T^p)^c}\mathbf{x}_{(T^p)^c} + \mathbf{n} - \mathbf{A}_{T^{q-p}}(\mathbf{x}_{T^{q-p}} + \bm{\delta}_{T^{q-p}})], \\
\mathbf{c} &\triangleq \mathbf{A}_{T^p}^\dagger (\mathbf{A}_{T^{q-p}}(\mathbf{x}_{T^{q-p}} + \bm{\delta}_{T^{q-p}})),
\end{align}
where we note that $\mathbf{a} = \mathbf{b} + \mathbf{c}$. According to Lemma \ref{RIP lemma 4},
\begin{align}
\label{OMP-CV e12} \Vert \mathbf{a}\Vert_2^2 &\leq (\frac{\delta_{k+1}}{1-\delta_p})^2 (\Vert\mathbf{x}_{(T^p)^c}\Vert_2^2 + \sigma_{\rm{n}}^2), \\
\label{OMP-CV e13} \Vert \mathbf{b}\Vert_2^2 &\leq (\frac{\delta_{k+1}}{1-\delta_p})^2 (\Vert \bm{\delta}_{T^{q-p}}\Vert_2^2 + \Vert\mathbf{x}_{(T^q)^c}\Vert_2^2 + \sigma_{\rm{n}}^2), \\
\label{OMP-CV e14} \Vert\mathbf{c}\Vert_2^2 &\leq (\frac{\delta_q}{1-\delta_p})^2 \Vert \mathbf{x}_{T^{q-p}} + \bm{\delta}_{T^{q-p}} \Vert_2^2.
\end{align}
Substituting (\ref{OMP-CV e10}) and (\ref{OMP-CV e11}) into (\ref{OMP-CV e8}) and (\ref{OMP-CV e9}), we obtain that,
\begin{gather}
\varepsilon_{\rm{g}}^p = \sigma_{\rm{n}}^2 + \Vert \mathbf{x}_{(T^q)^c}\Vert_2^2 + \Vert \mathbf{x}_{T^{q-p}} \Vert_2^2 + \Vert \mathbf{a} \Vert_2^2, \\
\varepsilon_{\rm{g}}^q = \sigma_{\rm{n}}^2 + \Vert \mathbf{x}_{(T^q)^c}\Vert_2^2 + \Vert \bm{\delta}_{T^{q-p}} \Vert_2^2 + \Vert \mathbf{b}\Vert_2^2.
\end{gather}
Furthermore, by inequalities (\ref{OMP-CV e12}), (\ref{OMP-CV e13}), and (\ref{OMP-CV e14}), one has,
\begin{equation} \label{OMP-CV e1}
\begin{split}
\varepsilon_{\rm{g}}^p \varepsilon_{\rm{g}}^q - \langle \Delta \mathbf{x}_{\rm{g}}^p, \Delta \mathbf{x}_{\rm{g}}^q \rangle ^2 = & (\Vert \mathbf{x}_{T^{q-p}}\Vert_2^2 + \Vert \mathbf{a}\Vert_2^2) (\Vert \bm{\delta}_{T^{q-p}}\Vert_2^2 + \Vert \mathbf{b}\Vert_2^2) - (\bm{\delta}_{T^{q-p}} \cdot \mathbf{x}_{T^{q-p}} - \mathbf{a} \cdot \mathbf{b})_2^2 
\\ & + (\Vert \mathbf{x}_{T^{q-p}}
+ \bm{\delta}_{T^{q-p}}\Vert_2^2 + \Vert \mathbf{a} - \mathbf{b} \Vert_2^2) (\sigma_{\rm{n}}^2 + \Vert\mathbf{x}_{(T^q)^c}\Vert_2^2)
\\ & \leq (\Vert \mathbf{x}_{T^{q-p}}\Vert_2^2 + \Vert \mathbf{a}\Vert_2^2) (\Vert \bm{\delta}_{T^{q-p}}\Vert_2^2 + \Vert \mathbf{b}\Vert_2^2) 
\\ & + (\Vert \mathbf{x}_{T^{q-p}}
+ \bm{\delta}_{T^{q-p}}\Vert_2^2 + \Vert \mathbf{a} - \mathbf{b} \Vert_2^2) (\sigma_{\rm{n}}^2 + \Vert\mathbf{x}_{(T^q)^c}\Vert_2^2)
\\ &\leq \left[\Vert \mathbf{x}_{T^{q-p}}\Vert_2^2 + (\frac{\delta_{k+1}}{1-\delta_p})^2 (\Vert\mathbf{x}_{(T^p)^c}\Vert_2^2 + \sigma_{\rm{n}}^2)\right] \\
~~&\left[ \Vert \bm{\delta}_{T^{q-p}}\Vert_2^2 + (\frac{\delta_{k+1}}{1-\delta_p})^2 (\Vert \bm{\delta}_{T^{q-p}}\Vert_2^2 + \Vert\mathbf{x}_{(T^q)^c}\Vert_2^2 + \sigma_{\rm{n}}^2)\right] 
\\ &+ (1 + (\frac{\delta_q}{1-\delta_p})^2)\Vert \mathbf{x}_{T^{q-p}} + \bm{\delta}_{T^{q-p}} \Vert_2^2 (\sigma_{\rm{n}}^2 + \Vert\mathbf{x}_{(T^q)^c}\Vert_2^2).
\end{split}
\end{equation}
Meanwhile, by combining (\ref{OMP-CV e8}) and (\ref{OMP-CV e9}),
\begin{equation} \label{OMP-CV e2}
\varepsilon_{\rm{g}}^p - \varepsilon_{\rm{g}}^q = \Vert \mathbf{x}_{T^{q-p}} \Vert_2^2 - \Vert \bm{\delta}_{T^{q-p}}\Vert_2^2 + \Vert \mathbf{a}\Vert_2^2 - \Vert \mathbf{b}\Vert_2^2.
\end{equation}

\subsection{Step3: Analysing Both Situations in Theorem \ref{OMP-CV theorem 1}}
So far, we have represented $\lambda$ using $\mathbf{\hat{x}}^p$ and $\mathbf{\hat{x}}^q$. Next, we analyse the CV comparison success probability in both situations of Theorem \ref{OMP-CV theorem 1} respectively. 

Let $\mathbf{\hat{x}}^o$ be the oracle output.

\noindent \textbf{Situation 1}: $T \subset T^o$ and  $T \backslash T^p \neq \emptyset$

In this situation, $p < o$ and $\Vert \mathbf{x}_{(T^o)^c}\Vert_2^2 = 0$. Consider recovered signals $\mathbf{\hat{x}}^p$ and $\mathbf{\hat{x}}^o$. On one hand, write $q$ as $o$ in Equation (\ref{OMP-CV e1}) we have, 
\begin{equation}
\begin{split}
\varepsilon_{\rm{g}}^p \varepsilon_{\rm{g}}^o - \langle \Delta \mathbf{x}_{\rm{g}}^p, \Delta \mathbf{x}_{\rm{g}}^o \rangle ^2 
\leq&  \left(1 + \left(\frac{\delta_{k+1}}{1-\delta_p}\right)^2\right)^2 \Vert \mathbf{x}_{T^{o-p}}\Vert_2^2 \Vert \bm{\delta}_{T^{o-p}}\Vert_2^2 \\
&\quad+ \left(1 + \left(\frac{\delta_{k+1}}{1-\delta_p}\right)^2\right)\left(\frac{\delta_{k+1}}{1-\delta_p}\right)^2\sigma_{\rm{n}}^2\left(\Vert\mathbf{x}_{T^{o-p}}\Vert_2^2 + \Vert\bm{\delta}_{T^{o-p}}\Vert_2^2\right) 
\\ & \quad +\left(\frac{\delta_{k+1}}{1-\delta_p}\right)^4 \sigma_{\rm{n}} ^4 + \left(1 + \left(\frac{\delta_{o}}{1-\delta_p}\right)^2\right) \Vert \mathbf{x}_{T^{o-p}} + \bm{\delta}_{T^{o-p}} \Vert_2^2 \sigma_{\rm{n}}^2.
\end{split}
\end{equation}
Then by Lemma \ref{OMP-CV lemma 2}, 
\begin{equation}
\begin{split}
\varepsilon_{\rm{g}}^p \varepsilon_{\rm{g}}^o - \langle \Delta \mathbf{x}_{\rm{g}}^p, \Delta \mathbf{x}_{\rm{g}}^o \rangle ^2 
&\leq  \left(1 + \left(\frac{\delta_{k+1}}{1-\delta_p}\right)^2\right)^2 \eta \Vert \mathbf{x}_{T^{o-p}}\Vert_2^2 \sigma_{\rm{n}}^2 \\
&\quad+ \left(1 + \left(\frac{\delta_{k+1}}{1-\delta_p}\right)^2\right)\left(\frac{\delta_{k+1}}{1-\delta_p}\right)^2\sigma_{\rm{n}}^2\left(\Vert\mathbf{x}_{T^{o-p}}\Vert_2^2 + \eta \sigma_{\rm{n}}^2\right) 
\\ &\quad + \left(\frac{\delta_{k+1}}{1-\delta_p}\right)^4 \sigma_{\rm{n}} ^4 + \left(1 + \left(\frac{\delta_{o}}{1-\delta_p}\right)^2\right) \left(\Vert \mathbf{x}_{T^{o-p}}\Vert_2^2 + \eta \sigma_{\rm{n}}^2\right) \sigma_{\rm{n}}^2.
\end{split}
\end{equation}
Recall that, 
\begin{equation}
\alpha^p = \frac{\Vert \mathbf{x}_{T \backslash T^p}\Vert_2}{\sigma_{\rm{n}}} = \frac{\Vert \mathbf{x}_{T^{o-p}}\Vert_2}{\sigma_{\rm{n}}},
\end{equation}
then, 
\begin{equation}\label{OMP-CV e4}
\varepsilon_{\rm{g}}^p \varepsilon_{\rm{g}}^o - \langle\Delta \mathbf{x}_{\rm{g}}^p, \Delta \mathbf{x}_{\rm{g}}^o \rangle ^2 \leq \frac{\beta_1}{2} (\alpha^p)^2 \sigma_{\rm{n}}^4 + \frac{\beta_2}{2} \sigma_{\rm{n}}^4,
\end{equation}
where 
\begin{align*}
\beta_1 &= 2 \left[\left(1+ \left(\frac{\delta_{k+1}}{1-\delta_p}\right)^2\right)^2\eta + \left(1+\left(\frac{\delta_{k+1}}{1-\delta_p}\right)^2\right)\left(\frac{\delta_{k+1}}{1-\delta_p}\right)^2 + \left(1+\left(\frac{\delta_o}{1-\delta_p}\right)^2\right)\right],\\
\beta_2 &= 2\left[\left(1+ \left(\frac{\delta_{k+1}}{1-\delta_p}\right)^2\right)\left(\frac{\delta_{k+1}}{1-\delta_p}\right)^2\eta + \left(\frac{\delta_{k+1}}{1-\delta_p}\right)^4 + \left(1+\left(\frac{\delta_o}{1-\delta_p}\right)^2\right) \eta\right].
\end{align*}

On the other hand, write $q$ as $o$ in (\ref{OMP-CV e2}) we have, 
\begin{equation}
\begin{split}
\varepsilon_{\rm{g}}^p - \varepsilon_{\rm{g}}^o & = \Vert \mathbf{x}_{T^{o-p}}\Vert_2^2 - \Vert \bm{\delta}_{T^{o-p}}\Vert_2^2 + \Vert \mathbf{a} \Vert_2^2 - \Vert \mathbf{b}\Vert_2^2
\\ & = (\Vert \mathbf{x}_{T^{o-p}}\Vert_2^2 + \Vert \mathbf{A}_{T^p}^\dagger \mathbf{A}_{T^{o-p}}\mathbf{x}_{T^{o-p}}\Vert_2^2) - (\Vert \bm{\delta}_{T^{o-p}}\Vert_2^2 + \Vert\mathbf{A}_{T^p}^\dagger \mathbf{A}_{T^{o-p}}\bm{\delta}_{T^{o-p}}\Vert_2^2 ) 
\\ & \quad + 2 \langle \mathbf{A}_{T^p}^\dagger \mathbf{A}_{T^{o-p}}(\mathbf{x}_{T^{o-p}} + \bm{\delta}_{T^{o-p}}),\mathbf{A}_{T^p}^\dagger \mathbf{n}\rangle
\\ & \geq \Vert \mathbf{x}_{T^{o-p}}\Vert_2^2 - (1+(\frac{\delta_o}{1-\delta_p})^2)\Vert \bm{\delta}_{T^{o-p}}\Vert_2^2 - 2 \frac{\delta_o \delta_{p+1}}{(1-\delta_p)^2} \Vert \mathbf{x}_{T^{o-p}} + \bm{\delta}_{T^{o-p}}\Vert_2 \sigma_{\rm{n}}
\\ & \geq (\alpha^p)^2 \sigma_{\rm{n}}^2 - (1+(\frac{\delta_o}{1-\delta_p})^2) \eta \sigma_{\rm{n}}^2 - 2 \frac{\delta_o \delta_{p+1}}{(1-\delta_p)^2} ((\alpha^p)^2 \sigma_{\rm{n}} + \sqrt{\eta} \sigma_{\rm{n}}) \sigma_{\rm{n}}
\\ & \geq \left[ (\alpha^p)^2 - \beta_3 \alpha^p - \beta_4 \right] \sigma_{\rm{n}}^2,
\end{split}
\end{equation}
where $\beta_3 = 2 \frac{\delta_o \delta_{p+1}}{(1-\delta_p)^2}$ and $\beta_4 = (1+(\frac{\delta_o}{1-\delta_p})^2)\eta + 2\frac{\delta_o \delta_{p+1}}{(1-\delta_p)^2} \sqrt{\eta}$. Notice that $\varepsilon_{\rm{g}}^p - \varepsilon_{\rm{g}}^o > 0$, then we have,
\begin{equation}\label{OMP-CV e5}
\varepsilon_{\rm{g}}^p - \varepsilon_{\rm{g}}^o \geq \max(\left[ (\alpha^p)^2 - \beta_3 \alpha^p - \beta_4 \right] \sigma_{\rm{n}}^2, 0).
\end{equation}
Combining equations (\ref{OMP-CV e4}) and (\ref{OMP-CV e5}) we finally reach 
\begin{equation} \label{OMP-CV e15}
\frac{\varepsilon_{\rm{g}}^p \varepsilon_{\rm{g}}^o - \langle \Delta \mathbf{x}_{\rm{g}}^p, \Delta \mathbf{x}_{\rm{g}}^o \rangle ^2}{(\varepsilon_{\rm{g}}^p - \varepsilon_{\rm{g}}^o)^2} \leq \frac{\frac{\beta_1}{2} (\alpha^p)^2 + \frac{\beta_2}{2}}{ \max(\left[ (\alpha^p)^2 - \beta_3 \alpha^p - \beta_4 \right], 0)^2}.
\end{equation}
Substituting (\ref{OMP-CV e15}) into (\ref{OMP-CV e3}),
\begin{equation}
\frac{1}{\lambda ^2} \leq \frac{2}{m_{\rm{cv}}} \left[ 1 + \frac{\beta_1 (\alpha^p)^2 + \beta_2}{\max(\left[ (\alpha^p)^2 - \beta_3 \alpha^p - \beta_4 \right], 0)^2} \right].
\end{equation}
Furthermore, 
\begin{equation}
\begin{split}\label{OMP-CV e6}
\lambda^2 &\geq \frac{m_{\rm{cv}}}{2} \left[ 1 - \frac{\beta_1 (\alpha^p)^2 + \beta_2}{\max(\left[ (\alpha^p)^2 - \beta_3 \alpha^p - \beta_4 \right], 0)^2 + (\beta_1 (\alpha^p)^2 + \beta_2)} \right]
\\ & = \frac{m_{\rm{cv}}}{2}[1 - \rm{g}(\alpha^p) ]
\\ & \approx \frac{m_{\rm{cv}}}{2} \left[ 1 - \frac{\beta_1}{(\alpha^p)^2 + \beta_1} \right],
\end{split}
\end{equation}
where $\rm{g}(\alpha^p) = [ \beta_1 (\alpha^p)^2 + \beta_2 ] / [\beta_1 (\alpha^p)^2 + \beta_2 + \max( (\alpha^p)^2 - \beta_3 \alpha^p - \beta_4, 0)^2 ]$.

The last approximation holds because $\beta_1 \gg \beta_2, \beta_3, \beta_4$. According to Theorem \ref{cv theorem 2}, $\Phi(\lambda)$ is the probability that $\epsilon_{\rm{cv}}^o < \epsilon_{\rm{cv}}^p$. Therefore, taking square root of both sides of (\ref{OMP-CV e6}) to prove the first part of Theorem \ref{OMP-CV theorem 1}. \\~

\noindent \textbf{Situation 2}: $T \subset T^o$ and  $T \backslash T^p = \emptyset$

In this situation, please notice that $\rho_{\rm{g}}$ of $\mathbf{\hat{x}}^o$ and $\mathbf{\hat{x}}^p$ is extremely closed to 1. Therefore, we would like to first calculate the lower bound of $\rho_{\rm{g}}$ and then complete the proof using Theorem \ref{cv theorem 3}. As $p$ may exceed $o$, we first assume $p < o$ and then analyze the alternative situation. Recall that 
\begin{equation}
\rho_{\rm{g}} = \frac{\langle\Delta \mathbf{x}_{\rm{g}}^o, \Delta \mathbf{x}_{\rm{g}}^p\rangle}{\Vert \Delta \mathbf{x}_{\rm{g}}^o\Vert_2 \Vert\Delta \mathbf{x}_{\rm{g}}^p\Vert_2}.
\end{equation}
On one hand,
\begin{equation}
\begin{split}
\langle\Delta \mathbf{x}_{\rm{g}}^o, \Delta \mathbf{x}_{\rm{g}}^p\rangle & = \langle \mathbf{A}_{T^p}^\dagger \mathbf{n}, \mathbf{A}_{T^p}^\dagger (\mathbf{n} - \mathbf{A}_{T^{o-p}} \bm{\delta}_{T^{o-p}})\rangle + \sigma_{\rm{n}}^2
\\ & = \Vert \mathbf{A}_{T^p}^\dagger \mathbf{n}\Vert_2^2 - \langle \mathbf{A}_{T^p}^\dagger \mathbf{n}, \mathbf{A}_{T^p}^\dagger \mathbf{A}_{T^{o-p}} \bm{\delta}_{T^{o-p}}\rangle + \sigma_{\rm{n}}^2
\\ & \geq \Vert \mathbf{A}_{T^p}^\dagger \mathbf{n}\Vert_2^2 - \Vert \mathbf{A}_{T^p}^\dagger \mathbf{n}\Vert_2 \Vert \mathbf{A}_{T^p}^\dagger \mathbf{A}_{T^{o-p}} \bm{\delta}_{T^{o-p}}\Vert_2 + \sigma_{\rm{n}}^2
\\ & \geq - \Vert \mathbf{A}_{T^p}^\dagger \mathbf{n}\Vert_2 \Vert \mathbf{A}_{T^p}^\dagger \mathbf{A}_{T^{o-p}} \bm{\delta}_{T^{o-p}}\Vert_2 + \sigma_{\rm{n}}^2
\\ & \geq \sigma_{\rm{n}}^2 (1 - \frac{\delta_{p+1} \delta_{o+1}}{(1-\delta_p)^2}\sqrt{\eta}),
\end{split}
\end{equation}
where the last inequality holds due to Lemma \ref{RIP lemma 4}. On the other hand,
\begin{equation}
\begin{split}
\varepsilon_{\rm{g}}^p \varepsilon_{\rm{g}}^o & = (\Vert \mathbf{A}_{T^p}^\dagger \mathbf{n}\Vert_2^2 + \sigma_{\rm{n}}^2) (\Vert \mathbf{A}_{T^p}^\dagger (\mathbf{n} - \mathbf{A}_{T^{o-p}} \bm{\delta}_{T^{o-p}})\Vert_2^2 + \Vert \bm{\delta}_{T^{o-p}}\Vert_2^2 + \sigma_{\rm{n}}^2 ) 
\\ & \leq \left(\left(\frac{\delta_{p+1}}{1-\delta_p}\right)^2 \sigma_{\rm{n}}^2 + \sigma_{\rm{n}}^2\right)\left(\left(\frac{\delta_{o+1}}{1-\delta_p}\right)^2 \sigma_{\rm{n}}^2 + \eta \sigma_{\rm{n}}^2 + \sigma_{\rm{n}}^2\right)
\\ & = \sigma_{\rm{n}}^4 \left(\left(\frac{\delta_{p+1}}{1-\delta_p}\right)^2 + 1\right)\left(\left(\frac{\delta_{o+1}}{1-\delta_p}\right)^2 + \eta + 1\right).
\end{split}
\end{equation}
Then,
\begin{equation}
\rho_{\rm{g}} = \frac{\langle\Delta \mathbf{x}_{\rm{g}}^o, \Delta \mathbf{x}_{\rm{g}}^p\rangle}{\Vert \Delta \mathbf{x}_{\rm{g}}^o\Vert_2 \Vert\Delta \mathbf{x}_{\rm{g}}^p\Vert_2} \geq \beta_5,
\end{equation}
where $\beta_5 = \frac{(1 - \frac{\delta_{p+1} \delta_{o+1}}{(1-\delta_p)^2}\sqrt{\eta})}{\sqrt{((\frac{\delta_{p+1}}{1-\delta_p})^2 + 1)((\frac{\delta_{o+1}}{1-\delta_p})^2 + \eta + 1)}}$.

For the other situation, i.e. $o<p$, just switch the position of $o$ and $p$ to obtain the same lower bound of $\rho_{\rm{g}}$.

Furthermore, it follows from Theorem \ref{cv theorem 3} that the CV comparison success probability of $\mathbf{\hat{x}}^p$ and $\mathbf{\hat{x}}^o$ is equal or higher than $\Phi (\lambda_0)$ if 
\begin{equation}\label{OMP-CV e7}
\frac{\varepsilon_{\rm{g}}^p}{\varepsilon_{\rm{g}}^o} \geq 2 C_0 + 1 + 2 \sqrt{C_0^2 + C_0},
\end{equation}
where $C_0 = (1-\rho_{\rm{g}}^2 )\frac{\lambda_0^2}{m_{\rm{cv}} - 2 \lambda_0^2}\leq (1-\beta_5^2) \frac{\lambda_0^2}{m_{\rm{cv}} - 2 \lambda_0^2}$.

Let $C_1 = 2 C_0 + 1 + 2 \sqrt{C_0^2 + C_0}$, then (\ref{OMP-CV e7}) can be written as $\varepsilon_{\rm{g}}^p < C_1 \varepsilon_{\rm{g}}^o$. Assume that among all recovered signals that recovered all indices in the support $T$, there are $n$ signals whose recovery error is larger than $C_1 \varepsilon_{\rm{g}}^o$, i.e. satisfying (\ref{OMP-CV e7}). Let $S$ denote the set containing iteration indices of these $n$ recovered signals, i.e., 
\begin{equation}
S = \{ i \mid  \varepsilon_{\rm{g}}^i \geq C_1 \varepsilon_{\rm{g}}^o \}.
\end{equation}
The probability is smaller than $n[1-\Phi(\lambda_0)]$ that among CV residuals of these $n$ recovered signals, there exists one that is smaller than $\epsilon_{\rm{cv}}^o$, i.e.
\begin{equation}
\textrm{P}(\exists i \in S ~~ s.t. ~~ \epsilon_{\rm{cv}}^i \leqslant \epsilon_{\textrm{cv}}^o) < n[1-\Phi(\lambda_0)].
\end{equation}
Since this is the probability smaller than which it holds that a recovered signal with recovery error larger $C_1 \varepsilon_{\rm{g}}^o$ is the OMP-CV output, i.e. 
\begin{equation}
\textrm{P}(\{\mathbf{\hat{x}}^i~\textrm{is~the~OMP-CV~output}\} \bigcap \{i \in S\}) < \textrm{P}(\exists i \in S ~~ \textrm{s.t.} ~~ \epsilon_{\rm{cv}}^i \leqslant \epsilon_{\rm{cv}}^o), 
\end{equation}
and since $n$ does not exceed $(d-k)$, it holds with probability larger than $\{1-(d-k)[1-\Phi(\lambda_0)]\}$ that OMP-CV output has recovery error smaller than $C_1 \varepsilon_{\rm{g}}^o$. This proves the second part of Theorem \ref{OMP-CV theorem 1}.
\end{proof}
\section{proof of Lemma \ref{OMP-CV lemma 1}}
\label{app OMP-CV lemma 1}
\begin{proof}
Make the notation $\mathbf{\Gamma}_i = (\mathbf{A}_{T^i}' \mathbf{A}_{T^i})^{-1}$, $i=p, q$, and assume next $p < q$. $\mathbf{\Gamma}_q$ can be written in $\mathbf{\Gamma}_p$ as,
\begin{equation}
\begin{split}
\mathbf{\Gamma}_q &= \left[ \begin{array}{c c}
\mathbf{A}_{T^p}' \mathbf{A}_{T^p} & \mathbf{A}_{T^p}' \mathbf{A}_{T^{q-p}} \\
\mathbf{A}_{T^{q-p}}' \mathbf{A}_{T^p} & \mathbf{A}_{T^{q-p}}' \mathbf{A}_{T^{q-p}}
\end{array} \right] ^ {-1} \\
&= \left[ \begin{array}{c c}
\mathbf{\Gamma}_p + \mathbf{\Gamma}_p \bm{\alpha}_p\bm{\theta}_p\bm{\alpha}_p'\mathbf{\Gamma}_p & -\mathbf{\Gamma}_p \bm{\alpha}_p\bm{\theta}_p \\
-\bm{\theta}_p \bm{\alpha}_p'\mathbf{\Gamma}_p & \bm{\theta}_p
\end{array}\right],
\end{split}
\end{equation}
where $\bm{\alpha}_p = \mathbf{A}_{T^p}' \mathbf{A}_{T^{q-p}}$, $\bm{\theta}_p = (\mathbf{A}_{T^{q-p}}' P_{T^p}\mathbf{A}_{T^{q-p}})^{-1}$, and $P_{T^p} = I - \mathbf{A}_{T^p} \mathbf{A}_{T^p}^\dagger$. Next, write $\mathbf{A}_{T^q}^\dagger$ in $\mathbf{A}_{T^p}^\dagger$ as,
\begin{equation}
\begin{split}
\mathbf{A}_{T^q}^\dagger &= (\mathbf{A}_{T^q}' \mathbf{A}_{T^q})^{-1} \mathbf{A}_{T^q}' = \mathbf{\Gamma}_q \mathbf{A}_{T^q}' \\
& = \left[ \begin{array}{c c}
\mathbf{\Gamma}_p + \mathbf{\Gamma}_p \bm{\alpha}_p\bm{\theta}_p\bm{\alpha}_p'\mathbf{\Gamma}_p & -\mathbf{\Gamma}_p \bm{\alpha}_p\bm{\theta}_p \\
-\bm{\theta}_p \bm{\alpha}_p'\mathbf{\Gamma}_p & \bm{\theta}_p
\end{array}\right] 
\left[ \begin{array}{c}
\mathbf{A}_{T^p}' \\
\mathbf{A}_{T^{q-p}}'
\end{array}\right] \\
& = \left[ \begin{array}{c}
\mathbf{A}_{T^p}^\dagger - \mathbf{A}_{T^p}^\dagger \mathbf{A}_{T^{q-p}} \bm{\theta}_p \mathbf{A}_{T^{q-p}}' P_{T^p} \\
\bm{\theta}_p \mathbf{A}_{T^{q-p}}' P_{T^p}
\end{array}\right].
\end{split}
\end{equation}

Finally, we reach that,  
\begin{equation}
\begin{split}
\mathbf{\hat{x}}_{T^q}^q &= \mathbf{A}_{T^q}^\dagger \mathbf{y} = \mathbf{A}_{T^q}^\dagger (\mathbf{A}_{T^q}\mathbf{x}_{T^q} + \mathbf{A}_{(T^q)^c} \mathbf{x}_{(T^q)^c} + \mathbf{n}) \\
&= \left[ \begin{array}{c}
\mathbf{A}_{T^p}^\dagger - \mathbf{A}_{T^p}^\dagger \mathbf{A}_{T^{q-p}} \theta_p \mathbf{A}_{T^{q-p}}' P_{T^p} \\
\bm{\theta}_p \mathbf{A}_{T^{q-p}}' P_{T^p}
\end{array}\right] (\mathbf{A}_{(T^q)^c} \mathbf{x}_{(T^q)^c} + \mathbf{n}) + \mathbf{x}_{T^q} \\
&= \left[ \begin{array}{c}
\mathbf{A}_{T^p} ^ \dagger (\mathbf{A}_{(T^q)^c} \mathbf{x}_{(T^q)^c} + \mathbf{n} - \mathbf{A}_{T^{q-p}} \bm{\delta}_{T^{q-p}})\\ 
\bm{\delta}_{T^{q-p}}
\end{array} \right]
+ \mathbf{x}_{T^q},
\end{split}
\end{equation}
where $\bm{\delta}_{T^{q-p}} = (\mathbf{A}_{T^{q-p}}' P_{T^p} \mathbf{A}_{T^{q-p}})^{-1} \mathbf{A}_{T^{q-p}}' P_{T^p} (\mathbf{A}_{(T^q)^c} \mathbf{x}_{(T^q)^c} + \mathbf{n})$.
\end{proof}
\section{proof of Lemma \ref{OMP-CV lemma 2}}
\label{app OMP-CV lemma 2}
\begin{proof}
Let $\mathbf{u} \in \mathbb{R}^{q-p}$ be an arbitrary vector. By Lemma \ref{RIP lemma 1}, 
\begin{equation}\label{App e5}
\sqrt{1-\delta_{q-p}} \Vert P_{T^p} \mathbf{A}_{T^{q-p}} \mathbf{u}\Vert_2 \leq \Vert \mathbf{A}_{T^{q-p}}'
P_{T^p} \mathbf{A}_{T^{q-p}} \mathbf{u}\Vert_2 \leq \sqrt{1+\delta_{q-p}} \Vert P_{T^p} \mathbf{A}_{T^{q-p}} \mathbf{u}\Vert_2.
\end{equation}
Furthermore, by Lemma \ref{RIP lemma 5},
\begin{equation}\label{App e6}
\sqrt{1-\left(\frac{\delta_q}{1-\delta_q}\right)^2} \Vert \mathbf{A}_{T^{q-p}} \mathbf{u} \Vert_2 \leq \Vert P_{T^p} \mathbf{A}_{T^{q-p}} \mathbf{u}\Vert_2 \leq \Vert \mathbf{A}_{T^{q-p}} \mathbf{u} \Vert_2.
\end{equation}
By Lemma \ref{RIP lemma 1} once again,
\begin{equation}\label{App e7}
\sqrt{1-\delta_{q-p}} \Vert \mathbf{u}\Vert_2 \leq \Vert  \mathbf{A}_{T^{q-p}} \mathbf{u}\Vert_2 \leq \sqrt{1+\delta_{q-p}} \Vert \mathbf{u}\Vert_2.
\end{equation}
Combining inequalities (\ref{App e5}), (\ref{App e6}), and (\ref{App e7}) we obtain that,
\begin{equation}
(1-\delta_{q-p})^2 \left(1-\left(\frac{\delta_q}{1-\delta_q}\right)^2\right) \Vert \mathbf{u}\Vert_2^2 \leq \Vert \mathbf{A}_{T^{q-p}}'
P_{T^p} \mathbf{A}_{T^{q-p}} \mathbf{u}\Vert_2 \leq (1+\delta_{q-p})^2 \Vert \mathbf{u}\Vert_2^2.
\end{equation}

The above inequality shows that the singular values of matrix $\mathbf{A}_{T^{q-p}}'
P_{T^p} \mathbf{A}_{T^{q-p}}$ lie between $(1-\delta_{q-p}) \sqrt{(1-(\frac{\delta_q}{1-\delta_q})^2)}$ and $(1+\delta_{q-p})$. Hence, 
\begin{equation}
\frac{1}{(1+\delta_{q-p})^2} \Vert \mathbf{u}\Vert_2^2 \leq \Vert (\mathbf{A}_{T^{q-p}}'
P_{T^p} \mathbf{A}_{T^{q-p}})^{-1} \mathbf{u} \Vert_2^2 \leq \frac{1}{(1-\delta_{q-p})^2 (1-(\frac{\delta_q}{1-\delta_q})^2)} \Vert \mathbf{u}\Vert_2^2.
\end{equation}
Finally, 
\begin{equation} \label{lemma e1}
\begin{split}
&\Vert \mathbf{\delta}_{T^{q-p}}\Vert_2^2 \\
 =& \Vert (\mathbf{A}_{T^{q-p}}' P_{T^p} \mathbf{A}_{T^{q-p}})^{-1} \mathbf{A}_{T^{q-p}}' P_{T^p} (\mathbf{A}_{(T^q)^c} \mathbf{x}_{(T^q)^c} + \mathbf{n}) \Vert_2^2 \\
\leq& \frac{1}{(1-\delta_{q-p})^2 (1-(\frac{\delta_q}{1-\delta_q})^2)} \Vert \mathbf{A}_{T^{q-p}}' P_{T^p} (\mathbf{A}_{(T^q)^c} \mathbf{x}_{(T^q)^c} + \mathbf{n})\Vert_2^2 \\
=&\frac{1}{(1-\delta_{q-p})^2 (1-(\frac{\delta_q}{1-\delta_q})^2)} \Vert \mathbf{A}_{T^{q-p}}' (\mathbf{A}_{(T^q)^c} \mathbf{x}_{(T^q)^c} + \mathbf{n}) - \mathbf{A}_{T^{q-p}}' \mathbf{A}_{T^p} \mathbf{A}_{T^p}^\dagger (\mathbf{A}_{(T^q)^c} \mathbf{x}_{(T^q)^c} + \mathbf{n}) \Vert_2^2 \\
\leq& \frac{1}{(1-\delta_{q-p})^2 (1-(\frac{\delta_q}{1-\delta_q})^2)} \left( \Vert \mathbf{A}_{T^{q-p}}' (\mathbf{A}_{(T^q)^c} \mathbf{x}_{(T^q)^c} + \mathbf{n}) \Vert_2^2 + \Vert \mathbf{A}_{T^{q-p}}' \mathbf{A}_{T^p} \mathbf{A}_{T^p}^\dagger (\mathbf{A}_{(T^q)^c} \mathbf{x}_{(T^q)^c} + \mathbf{n}) \Vert_2^2 \right) \\
\leq& \frac{1}{(1-\delta_{q-p})^2 (1-(\frac{\delta_q}{1-\delta_q})^2)} \left(\delta_{\vert (T^q)^c \vert +q-p+1}^2 + \left(\frac{\delta_q \delta_{\vert (T^q)^c \vert +p+1}}{1-\delta_p}\right)^2\right)(\Vert \mathbf{x}_{(T^q)^c}\Vert_2^2 + \sigma_{\rm{n}}^2) \\
=& \eta (\Vert \mathbf{x}_{(T^q)^c}\Vert_2^2 + \sigma_{\rm{n}}^2),
\end{split}
\end{equation}
where the second inequality holds due to Cauchy-Schwarz Inequality and 
$$
\eta = \frac{1}{(1-\delta_{q-p})^2 \left(1-\left(\frac{\delta_q}{1-\delta_q}\right)^2\right)} \left(\delta_{\vert (T^q)^c \vert +q-p+1}^2 + \left(\frac{\delta_q \delta_{\vert (T^q)^c \vert +p+1}}{1-\delta_p}\right)^2\right).
$$ 

If, e.g., $\delta_d \leq 0.1$, since all footnotes of RICs in the above equation do not exceed $d$, one has, 
\begin{equation}
\eta \leq 0.0127.
\end{equation}
\end{proof}
\bibliographystyle{ieeetr}
\bibliography{ref}

\end{document}